%% file: main.tex
\PassOptionsToPackage{prologue,dvipsnames}{xcolor}
\documentclass[acmsmall]{acmart}

\setcopyright{cc}
\setcctype{by}
\acmDOI{10.1145/3763072}
\acmYear{2025}
\acmJournal{PACMPL}
\acmVolume{9}
\acmNumber{OOPSLA2}
\acmArticle{294}
\acmMonth{10}
\received{2025-03-24}
\received[accepted]{2025-08-12}
\usepackage{booktabs}   
\usepackage{subcaption} 
\usepackage{algorithm}
\usepackage{algorithmicx}
\usepackage[noend]{algpseudocode}
\usepackage{amsmath}
\usepackage{stmaryrd}
\usepackage{mdframed}
\usepackage{url}
\usepackage{listings}
\usepackage[bottom]{footmisc}
\usepackage{graphicx}
\usepackage{wrapfig}
\usepackage{fancyvrb}
\usepackage{paralist}
\usepackage{caption}
\usepackage{blindtext}
\usepackage{comment}
\usepackage{xspace}
\usepackage{multirow}
\usepackage{multicol}
\usepackage{booktabs}
\usepackage{enumerate}
\usepackage{enumitem}
\usepackage{caption}
\usepackage{graphicx}
\usepackage{listing}
\usepackage{tikz}
\usepackage{ifsym}
\usepackage{circledsteps}
\usepackage[frozencache,cachedir=minted-cache]{minted}
\usepackage{soul}

\usepackage{appendix}
\usepackage[nameinlink]{cleveref}
\usepackage{textgreek}

\usepackage{makecell}
\usepackage{tcolorbox}
\usepackage{sidecap}

\usepackage{mathtools, stmaryrd}

\usepackage[utf8]{inputenc}
\usepackage{pgfplots}
\input{macros}
\begin{document}
\title{Homomorphism Calculus for User-Defined Aggregations}
%
%
\author{Ziteng Wang}
\orcid{0009-0001-8487-8093}
\affiliation{%
  \institution{University of Texas at Austin}
  \city{Austin}
  \country{USA}
}
\email{ziteng@utexas.edu}

\author{Ruijie Fang}
\orcid{0000-0001-5348-5468}
\affiliation{%
  \institution{University of Texas at Austin}
  \city{Austin}
  \country{USA}
}
\email{ruijief@cs.utexas.edu}

\author{Linus Zheng}
\orcid{0009-0005-9832-4066}
\affiliation{%
  \institution{University of Texas at Austin}
  \city{Austin}
  \country{USA}
}
\email{linusjz@utexas.edu}

\author{Dixin Tang}
\orcid{0000-0002-3316-6651}
\affiliation{%
  \institution{University of Texas at Austin}
  \city{Austin}
  \country{USA}
}
\email{dixin@utexas.edu}

\author{Işıl Dillig}
\orcid{0000-0001-8006-1230}
\affiliation{%
  \institution{University of Texas at Austin}
  \city{Austin}
  \country{USA}
}
\email{isil@cs.utexas.edu}

\ccsdesc[500]{Software and its engineering~Automatic programming}
\ccsdesc[500]{Computing methodologies~Parallel computing methodologies}

\input{abstract}
\maketitle              
\input{intro}

\input{overview}

\input{problem}

\input{methodology}

\input{impl}

\input{eval}
\input{limitations}
\input{related}

\input{conclusion}
\input{data-avail}
\input{acks}
%
%
%
%
%
\bibliography{main}

\newpage
\appendix
\input{proof}

\input{appendix}

\end{document}

%% file: macros.tex
\newtheorem{problem}{Problem}

\newcommand{\todo}[1]{{{#1}}}

\newcommand{\prog}{\mathcal{P}}
\newcommand{\dftrans}{\Phi}
\newcommand{\proj}{\small \mathsf{project}}
\newcommand{\select}{\small \mathsf{select}}

\newcommand{\map}{{\small \mathsf{map}}}
\newcommand{\filter}{{\small \mathsf{filter}}}
\newcommand{\zip}{\small \mathsf{zip}}
\newcommand{\ite}{\small \mathsf{ITE}}

\newcommand{\fold}{\small \mathsf{fold}}

\newcommand{\append}{\small \mathsf{append}}
\newcommand{\update}{\small \mathsf{update}}
\newcommand{\union}{\small \mathsf{union}}
\newcommand{\setinsert}{\small \mathsf{insert}}

\newcommand{\default}[1]{\Delta(#1)}
\newcommand{\convert}{\digamma}

\newcommand{\Synduce}{{\sc Synduce}\xspace}
\newcommand{\AutoLifter}{{\sc AutoLifter}\xspace}
\newcommand{\SuFu}{{\sc SuFu}\xspace}
\newcommand{\ParSynt}{{\sc ParSynt}\xspace}

\newcommand{\vbooltrue}{\mathsf{true}}
\newcommand{\vboolfalse}{\mathsf{false}}

\newcommand{\isf}{\Lambda}
\newcommand{\collectionjoin}{\boxtimes}
\newcommand{\sem}[1]{\llbracket#1\rrbracket}

\newcommand{\TInt}{\mathsf{Int}}
\newcommand{\TChar}{\mathsf{Char}}

\newcommand{\TFloat}{\mathsf{Float}}

\newcommand{\TMap}[2]{\mathsf{Map}\langle#1,#2\rangle}

\newcommand{\TList}[1]{\mathsf{List}\langle #1 \rangle}

\newcommand{\CITE}[3]{\mathsf{ITE}(#1,#2,#3)}
\newcommand{\Cfst}[1]{\mathsf{fst}(#1)}
\newcommand{\Csnd}[1]{\mathsf{snd}(#1)}

\newcommand{\CListGet}[2]{\mathsf{get}(#1,#2)}

\newcommand{\CListFoldl}[3]{\mathsf{foldl}(#1,#2,#3)}
\newcommand{\CListMap}[2]{\mathsf{map}(#1,#2)}

\newcommand{\CNull}{\mathtt{null}}

\crefname{problem}{problem}{problems}
\crefname{lemma}{lemma}{lemmas}
\crefname{claim}{claim}{claims}

\newcommand{\propagateparamto}{\twoheadrightarrow}
\newcommand{\vertdecompto}{\rightsquigarrow}
\newcommand{\horizdecompto}{\hookrightarrow}
\newcommand\normalizes[3]{
    \ensuremath{\left(#1, #2\right) \sim #3}
}

\newcommand{\normof}[3]{
    \mathsf{Norm}\left( #1, #2, #3 \right)
}

\newcommand{\bounded}[2]{
    \llparenthesis {#2} \rrparenthesis
}
\newcommand{\unbounded}[4]{
    {#1}\llbracket {#2}, {#3}, {#4} \rrbracket
}
\newcommand{\unboundedapp}[4]{
    {#1}_{#4}\llbracket {#2}, {#3} \rrbracket
}
\newcommand{\vdcomp}[2]{{#1}\funccomp{#2}}
\newcommand{\vdid}{\funcid}
\newcommand{\predtrue}{\top}
\newcommand{\vdpred}{\phi}
\newcommand{\typedestructor}{d}
\newcommand{\decomp}{\Omega}

\newcommand{\sidecond}[1]{{\color{blue}\left[\text{\small #1}\right]}}

\newcommand{\dtuple}{\decomp_T}
\newcommand{\dcollection}{\decomp_C}
\newcommand{\vdfunc}{\mathcal{T}}
\newcommand{\decomptoexpr}{\looparrowright}

\newcommand{\params}{\mathsf{BoundVars}}

\newcommand{\idiotbox}[2]{
\fbox{
\begin{minipage}{0.9\linewidth}
{\bf Result for {#1}:}
{#2}
\end{minipage}
}
}

\newcommand{\norefute}{\textsc{NoRefute}}
\newcommand{\noreduce}{\textsc{NoReduce}}
\newcommand{\nodecomp}{\textsc{NoDecomp}}
\newcommand{\nodeduce}{\textsc{NoDeduce}}

\newcommand{\dfconcat}{\boxplus}
\newcommand{\df}{\mathcal{D}}
\newcommand{\dftypee}[1]{\mathsf{DF}\langle #1 \rangle}
\newcommand{\dfcols}{\mathcal{C}}
\newcommand{\dfrows}{\mathcal{R}}
\newcommand{\dfvals}{\mathcal{V}}
\newcommand{\dftype}{\mathcal{T}}
\newcommand{\dfagg}{\mathcal{P}}
\newcommand{\dfmerge}{\oplus}
\newcommand{\udaf}{f}
\newcommand{\norm}{h}
\newcommand{\trans}{\Phi}
\newcommand{\agg}{\small \mathsf{aggregate}}

\newcommand{\dfbody}{\mathcal{F}}

\newcommand{\typeprod}{\type_p}

\newcommand{\typebase}{\type_b}
\newcommand{\typecollection}[1]{\type_c\langle{#1}\rangle}
\newcommand{\typeiter}{\textsf{Iter}}

\newcommand{\fv}{\textsf{fv}}
\newcommand{\abslambda}[2]{\lambda {#1}. ~{#2}}
\newcommand{\typeselector}{\sigma}
\newcommand{\type}{\tau}

\DeclareMathOperator{\concat}{+\!+}

\DeclareMathOperator{\fresh}{\textsf{freshVar}}
\DeclareMathOperator{\funcid}{\textsf{Id}}

\DeclareMathOperator{\funccomp}{\circ}

\DeclareMathOperator{\project}{\mathsf{project}}

\newcommand{\scalar}{\textit{Sc}}

\newcommand{\emptylist}{\texttt{[]}}
\newcommand{\emptymap}{\texttt{\{\}}}

\makeatletter
\newcommand{\pequiv}[3]{%
    \ifinalign@
        #1 &\equiv \left(#2, #3\right)
    \else
        #1 \equiv \left(#2, #3\right)
    \fi
}
\newcommand{\requiv}[3]{%
    \ifinalign@
        {#1} &\equiv_{#2} {#3}
    \else
        {#1} \equiv_{#2} {#3}
    \fi
}

\newcommand{\irule}[2]%
   {\mkern-2mu\displaystyle\frac{#1}{\vphantom{,}#2}\mkern-2mu}
\newcommand{\irulelabel}[3]
{
\mkern-2mu
\begin{array}{ll}
\displaystyle\frac{#1}{\vphantom{,}#2} & \!\!\!\! #3
\end{array}
\mkern-2mu
}

\newcommand{\bfpara}[1]{\vspace{5pt}\noindent\emph{\textbf{#1}}}

\newcommand{\range}[1]{\textsf{range}(#1)}
\newcommand{\set}[1]{\{ #1 \}}

\newcommand{\init}{\mathcal{I}}

\newcommand{\rfs}{\Phi}

\newcommand{\hole}{\Box}
\newcommand{\pp}{\nolinebreak{\bf +}\nolinebreak\hspace{-.10em}{\bf +}}

\newcommand{\tool}{\textsc{Ink}\xspace}
\newcommand{\dsl}{\textsc{Inkling}\xspace}
\newcommand{\toolname}{\tool}
\newcommand{\parsynt}{\textsc{Parsynt}}
\newcommand{\autolifter}{\textsc{AutoLifter}}

\newcommand{\spark}{\textsc{Spark}}
\newcommand{\flink}{\textsc{Flink}}

\lstdefinestyle{numbers}
{
  numbers=left,
  numberstyle=\sf,
  xleftmargin=15pt
}

\newmintinline[T]{scala}{fontsize=\footnotesize, escapeinside=||}
\newmintinline[Tnormal]{scala}{escapeinside=||}

\newminted[Scala]{scala}{fontsize=\footnotesize,
escapeinside=||,
numbersep=5pt}

\newcommand{\correctcode}[1]{{\color{ForestGreen} \textbf{#1}}}
\newcommand{\wrongcode}[1]{{\color{red} \textbf{#1}}}

\algnewcommand{\IIf}[1]{\State\algorithmicif\ #1\ \algorithmicthen}

%% file: abstract.tex
\begin{abstract}
Data processing frameworks like Apache Spark and Flink  provide built-in support for  user-defined aggregation functions (UDAFs),  enabling the integration of  domain-specific logic. However, for these frameworks to support \emph{efficient} UDAF execution, the function needs to satisfy a \emph{homomorphism property},  which ensures that partial results from independent computations can be merged correctly. Motivated by this problem, this paper introduces a novel \emph{homomorphism calculus} that can both verify and refute whether a UDAF is a dataframe homomorphism. If so, our calculus also enables the construction of a corresponding merge operator which can be used for incremental computation and parallel execution.  We have implemented an algorithm based on our proposed calculus and  evaluate it on real-world UDAFs, demonstrating  that our approach significantly outperforms two leading synthesizers.
\end{abstract}

%% file: intro.tex
\section{Introduction}\label{sec:intro}

\begin{wrapfigure}{r}{0.22\textwidth}
  \centering
  \vspace{-0.1in}
  \includegraphics[width=0.18\textwidth]{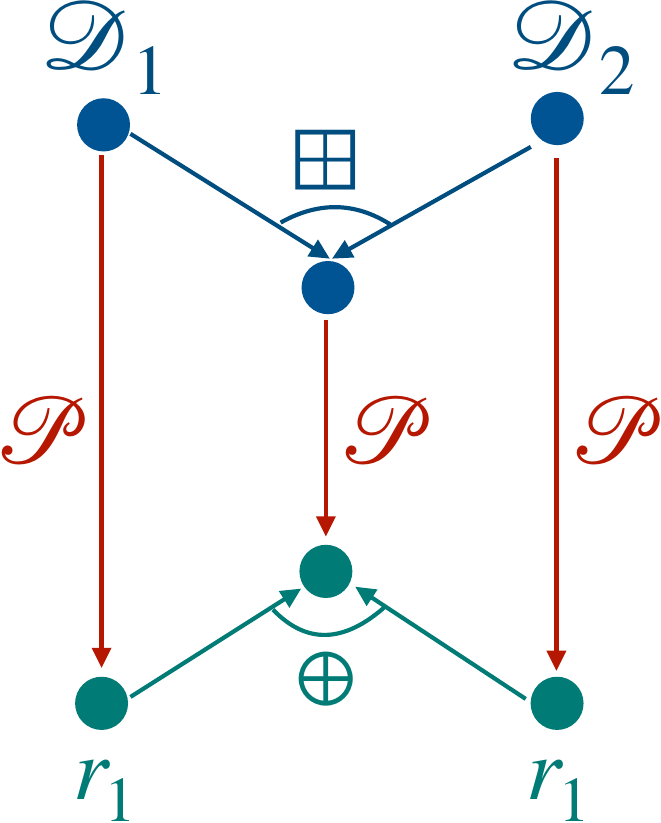}
  \vspace{-0.1in}
  \caption{\small Homomorphism}\label{fig:homo}
  \vspace{-0.1in}
\end{wrapfigure}

User-defined aggregation functions (UDAFs) are custom functions that allow users to perform complex aggregation operations. UDAFs play an important role in data science applications because they enable the implementation of domain-specific logic and complex calculations, such as custom statistical measures or weighted averages. Due to their growing importance, many frameworks such as Apache \spark{}~\cite{spark} and \flink{}~\cite{flink} provide extensive support for UDAFs, allowing users to perform complex aggregations over tabular data (henceforth referred to as \emph{dataframes}).

However, to execute UDAFs efficiently through parallel or incremental computation,  the function needs to be a \emph{homomorphism}. Intuitively, this means that the UDAF $\dfagg$ can be applied independently to two dataframes, \(\df_1\) and \(\df_2\), with the results subsequently combined using a binary \emph{merge operator}. Formally, as illustrated in \Cref{fig:homo},  $\dfagg$ is defined as a dataframe homomorphism if there exists a merge operator \(\dfmerge\) satisfying:
\[
\dfagg(\df_1 \dfconcat \df_2) = \dfagg(\df_1) \dfmerge \dfagg(\df_2),
\]
where $\dfconcat$ denotes dataframe concatenation.  If $\dfagg$ meets this criterion, $\dfmerge$ can be used to support both parallelization and incremental computation.

While prior work~\cite{parsynt_pldi17,
parsynt_pldi19,parsynt_pldi21}  has studied list homomorphisms for enabling parallel computation, existing techniques largely focus on computations over lists of scalars and use syntax-guided inductive synthesis to construct a suitable merge operator. However,  real-world UDAFs often involve complex inputs as well as intermediate data structures, such as maps and  nested collections, that are not straightforward to handle using prior techniques. 

Motivated by this shortcoming, this paper proposes a new \emph{homomorphism calculus} that can be used to verify whether or not a given function $\dfagg$ is a dataframe homomorphism.   If $\dfagg$ is proven to be a homomorphism,  our calculus also synthesizes the merge operator $\dfmerge$, thereby providing a constructive proof that can be leveraged {by the query optimizer 
of the underlying data processing framework.}
Notably, our approach to merge operator construction is largely based on \emph{deductive reasoning} and only resorts to search-based inductive synthesis for simple problems involving scalar values rather than complex data types.

The central idea of our homomorphism calculus is that a dataframe aggregation qualifies as a homomorphism if and only if its {\it accumulator function}---the component responsible for 
processing individual rows---satisfies a specific commutativity condition. We formalize this condition using right and left actions on a set, introducing a \emph{normalizer} function that underpins our calculus. Specifically, we demonstrate that a program is a homomorphism if and only if an appropriate normalizer exists for the accumulator function. This insight transforms the problem of synthesizing a merge operator for the entire aggregation into the simpler task of synthesizing a normalizer for the accumulator.




While reducing the synthesis problem from merge operators to normalizers makes it more manageable, constructing normalizers can still be difficult when the accumulator maintains complex internal state. 
Our calculus addresses this complexity through \emph{type-directed decomposition}.
For example, consider an aggregation operation with internal state of type $\TList{\tau}$. A na\"ive approach would require synthesizing a function that operates on two inputs of type $\TList{\tau}$, which becomes increasingly difficult as the complexity of the type parameter $\tau$ grows. To manage this complexity, our calculus reduces the synthesis problem for lists of type $\TList{\tau}$ to synthesis problems involving the element type $\tau$, continuing this decomposition until no further simplification is possible. 

\begin{figure}[t]
\includegraphics[scale=0.37]{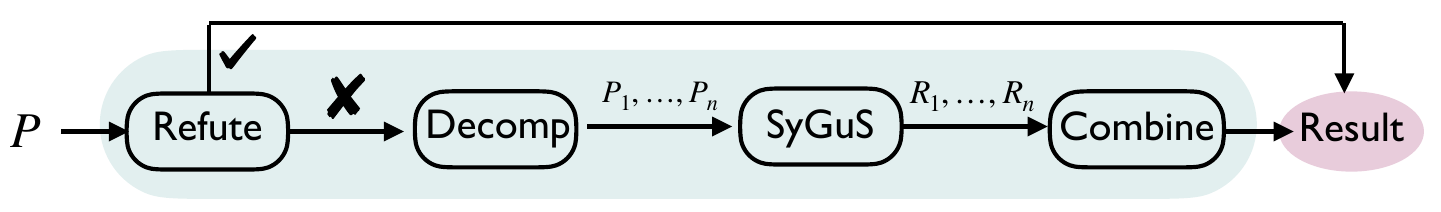}
\vspace{-0.1in}
\caption{Overview of our approach}\label{fig:overview}
\vspace{-0.2in}
\end{figure}


Another key aspect of our calculus is its ability to refute the homomorphism property. 
Since attempting to construct a normalizer when none exists can be highly inefficient,  refutation rules in our calculus help prevent futile synthesis efforts.
\Cref{fig:overview} gives an overview of our  verification algorithm that is based on the proposed homomorphism  calculus. Our method first attempts to refute the existence of a merge operator, and, if refutation fails, it decomposes the synthesis problem into simpler sub-problems. When further decomposition is not possible,  it resorts to syntax-guided inductive synthesis (SyGuS) for the leaf-level problems. Finally, it combines the solutions of these sub-problems using deductive synthesis. Assuming the existence of an oracle for solving these leaf-level SyGuS problems, the resulting procedure is both sound and complete.

We have implemented the proposed algorithm in a tool called \toolname{} and evaluated it on 50 real-world UDAFs targeting  Apache \spark{} and \flink{}. We compare our approach against two baselines, a state-of-the-art SyGuS solver, CVC5, and a synthesizer, \ParSynt, for divide-and-conquer parallelism. Our experimental results show that \toolname{} significantly outperforms these baselines in both synthesis and refutation tasks. Additionally, ablation studies demonstrate the impact of the core ideas underlying our approach.

To summarize, this paper makes the following key contributions:

\begin{itemize}[leftmargin=*]
\item  We present a  \emph{homomorphism calculus} for proving and refuting  homomorphisms and constructing a merge operator that enables parallel and incremental computation.
\item We prove that a program is a homomorphism if and only if its accumulator function satisfies a  generalized commutativity condition. We formalize this concept via \emph{normalizers} and show how  to simplify the problem using an alternative  specification. 
\item We show how to effectively tackle accumulators with complex internal state through a novel type-directed decomposition technique. 
\item We implement a  verification and synthesis procedure based on  the proposed homomorphism calculus in  a new  tool called \toolname{} and demonstrate its   effectiveness on real-world UDAFs.
\end{itemize}

%% file: overview.tex
\section{Overview}
\label{sec:overview}
\input{figures/motivating-example}

In this section, we outline our approach through an example illustrated in \Cref{fig:motivating-example}. This example features an Apache \spark{} program written in Scala that involves a custom user-defined aggregate function (UDAF). This program is designed to process bid data in a dataframe that contains three columns: \texttt{BidPrice}, \texttt{AuctionYear}, and \texttt{Item}. 
The result of the program is a tuple containing (1) the highest bid in 2024, (2) the number of bids above 1000 dollars in the same year, and (3) bid counts in 2024 for each item. 
\Cref{fig:io} displays a sample input-output pair for this program.

To understand what this program does, consider \spark{}'s aggregation mechanism, which involves two core phases: initialization and update. During the initialization phase, an aggregation buffer is set up with default starting values, such as zeros or empty collections. In the update phase, each row in the dataframe is processed sequentially,
with the buffer  modified to accumulate results.
In our running example, the buffer is initialized to \Tnormal{zero} in Figure~\ref{fig:motivating-example}, and the update logic is defined by the \Tnormal{reduce} function, which (a) updates the maximum bid
 if the current row’s bid price is higher, 
(b) increments the count of high bids  if the bid exceeds 1000, and (c) updates the item bid map  reflect the number of bids per item.

\begin{figure}
\centering
\small
\noindent\begin{minipage}[t]{0.43\textwidth}
\begin{tabular}{|c|c|c|}
    \hline 
    \textbf{BidPrice} : $\TFloat$ & \textbf{AuctionYear} : $\TInt$ & \textbf{Item} : $\TInt$ \\ 
    \hline 
    330.94 & 2024 & 3 \\ 
    1192.08 & 2024 & 2 \\
    161.11 & 2019 & 9 \\
    ... & ... & ... \\
    \hline
\end{tabular}
\end{minipage}%
\hfill%
\begin{minipage}[t]{0.43\textwidth}
\vspace*{-0.3in}
   \begin{equation*}
   \left(\underset{\text{Highest bid}}{\boxed{1192.08}}, 
   \overset{\text{\#bid price > 1000}}{\boxed{9}}, \underset{\text{\# of bids per item}}{\boxed{\{2 \mapsto 5, 3 \mapsto 2, \ldots\}}}\right)
   \end{equation*}
\end{minipage} 
\vspace{-0.05in}
\caption{Sample input (left) and output (right) of the \spark{} program. }\label{fig:io}
\end{figure}

\begin{figure}
\small
\centering
\begin{Scala}
  def merge(buffer1: BidAggBuffer, buffer2: BidAggBuffer): BidAggBuffer = {
    val mergedMaxBid = math.max(buffer1.maxBid, buffer2.maxBid)
    val mergedHighBidCount = buffer1.highBidCount + buffer2.highBidCount

    val map1 = buffer1.itemBidCounts
    val map2 = buffer2.itemBidCounts
    val mergedMap = map1 ++ map2.map { case (k, v) => k -> (v + map1.getOrElse(k, 0)) }

    BidAggBuffer(mergedMaxBid, mergedHighBidCount, mergedMap)
  }
\end{Scala}
\vspace*{-0.2in}
\caption{The merge function for the aggregation in \Cref{fig:motivating-example}.}
\vspace*{-0.2in}
\label{fig:motivating-example-merge}
\end{figure}

To support distributed and incremental processing, \spark{} must be able to combine intermediate results from different slices of the dataframe, but this can be done correctly only if the overall program defines a homomorphism. 
Going back to our example, this computation is indeed a homomorphism, so the user can take advantage of \spark{}'s distributed processing capabilities by implementing a suitable merge function, such as the one shown in \Cref{fig:motivating-example-merge}. This \texttt{merge} function combines two aggregation buffers, \texttt{buffer1} and \texttt{buffer2}, into a single result 
by updating the maximum bid price, aggregating the count of high bids over 1000, and merging per-item bid counts.
In the remainder of this section, we outline the key aspects of our approach that allow us to synthesize the \texttt{merge} function from \Cref{fig:motivating-example-merge}.

\vspace*{0.1in}
\begin{mdframed}
    \textbf{Idea 1:} An aggregation program defines a homomorphism if and only if its accumulator function has a so-called \emph{normalizer}. If so, the normalizer of the accumulator  corresponds to the desired merge function for the whole aggregation.

\end{mdframed}
\vspace*{0.1in}


\begin{wrapfigure}{r}
{0.35\textwidth}
    \centering
 \vspace{-0.2in}   \includegraphics[width=0.9\linewidth]{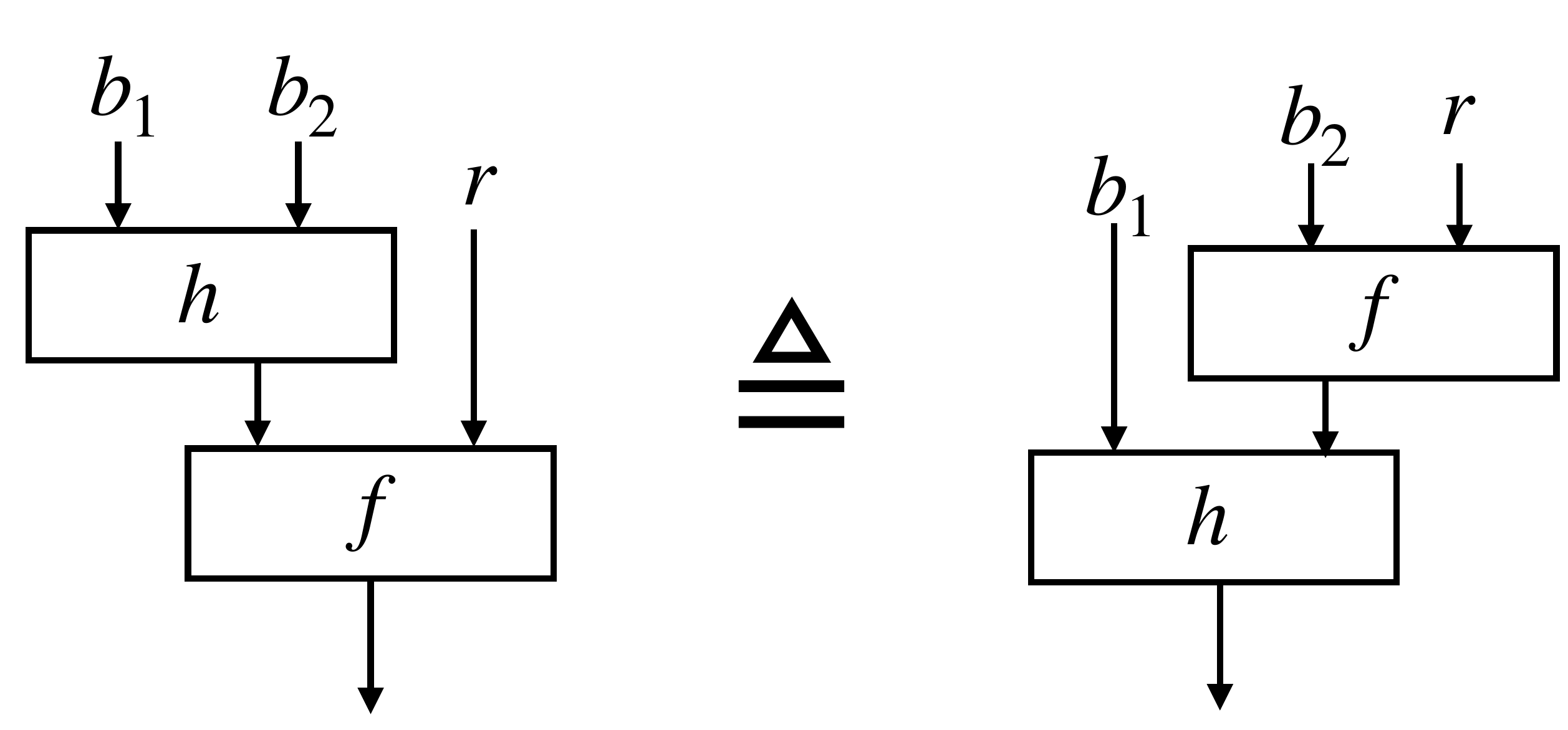}
 \vspace{-0.15in}
    \caption{Commutativity condition.}
    \label{fig:commutativity-motivation}
    \vspace*{-0.15in}
\end{wrapfigure}

The aggregation program in \Cref{fig:motivating-example} operates over the entire dataframe, whereas the accumulator function, represented by the \texttt{reduce} method, processes a single row at a time, making it easier to analyze. Fortunately, we can determine whether an aggregation defines a homomorphism by focusing solely on the accumulator and checking whether it admits a \emph{normalizer}. Intuitively, a normalizer $h$ for a function $f$ must satisfy a condition we refer to as \emph{generalized commutativity}, depicted in \Cref{fig:commutativity-motivation}, which corresponds to the following algebraic law:
\[
\forall b_1, b_2, r. ~~ f(h(b_1, b_2), r) = h(b_1, f(b_2, r))
\]
Here, $b_1$ and $b_2$ represent partial aggregation results, and $r$ denotes a new row in the dataframe. This law ensures that merging $b_1$ and $b_2$ and then applying \texttt{reduce} to the result is equivalent to first reducing $b_2$ with $r$ and then merging with $b_1$. We call this property \emph{generalized commutativity} because it formalizes how two \emph{distinct} functions—namely, the accumulator $f$ and the normalizer $h$ can commute. This differs from standard commutativity (which applies to a single binary operator) and associativity (which involves regrouping operands).

\vspace*{0.1in}
\begin{mdframed}
    \textbf{Idea 2:} We can greatly simplify the normalizer synthesis problem through decomposition.

\end{mdframed}
\vspace*{0.1in}

Our formulation so far simplifies the original problem in that the specification does not involve the entire dataframe. One obvious way to solve the resulting synthesis problem is to use syntax-guided synthesis (SyGuS) ~\cite{sygus} by  providing a suitable DSL in which the merge function can be expressed. However, it turns out that, for many real-world examples, directly synthesizing the merge function is quite challenging using existing SyGuS solvers. For instance,  the \texttt{merge} function shown in Figure~\ref{fig:motivating-example-merge} needs to correctly compute three different results {\Tnormal{maxBid, highBidCount, itemBidCounts}},
where  \Tnormal{itemBidCounts} is a mapping from integers to integers.  Thus, the merge operator needs to  iterate through the key-value pairs in \Tnormal{map2} and correctly update \Tnormal{map1} by summing counts for each key. This step involves both accessing and modifying an arbitrary number of entries, making it significantly more complex than merging simple numeric values.

\begin{table}[t]
\caption{Original aggregation expressions and their merge expressions.} \label{fig:motivating-decomp}
\vspace{-0.10in}
\begin{tabular}{|l|l|}
\hline

\textbf{Original expression} & \textbf{Merge expression} \\ \hline

\T{|$e_1 =$| math.max(buffer.maxBid, data.bidPrice)} &
{\footnotesize $h_1=$} \T{(s1, s2) => math.max(s1, s2)} \\ \hline

\makecell[l]{
{\footnotesize $e_2 =$} \T{if (data.bidPrice > 1000) buffer.highBidCount + 1} \\
\phantom{{\footnotesize $e_2 =$}} \T{else buffer.highBidCount}} &
{\footnotesize $h_2=$} \T{(s1, s2) => s1 + s2} \\ \hline

\T{|$e_3 =$| itemBidCountMap.getOrElse(item, 0) + 1} &
{\footnotesize $h_3=$} \T{(v1, v2) => v1 + v2} \\ \hline

\T{|$e_4 =$| itemBidCountMap + (item -> |$e_3$|)} &
\makecell[l]{
{\footnotesize $h_4 =$}           \T{(m1, m2) => m1 ++ m2.map (} \\
\phantom{{\footnotesize $f_4 =$}} \T{case (k, v) => } \\
\phantom{{\footnotesize $f_4 =$}} \T{  k -> |$h_3$|(v, m1.getOrElse(k, 0)))}
} \\ \hline

{\footnotesize \begin{tabular}{@{}r@{ }l@{}}
     $\udaf=$ & \T{(buffer, data) => (|$e_1$|, |$e_2$|, |$e_4$|)} \\
     $\dfagg=$ & \T{bidsDF} \\
     \phantom{$\dfagg=$} & \ \ \T{.filter(year(col("AuctionYear")) == 2024)} \\
     \phantom{$\dfagg=$} & \ \ \T{.select("BidPrice", "Item").as[BidData]} \\
     \phantom{$\dfagg=$} & \ \ \T{.aggregate(|$\udaf$|, initializer)} \\
\end{tabular}} &
\makecell[l]{
\T{((a1, a2, a3), (b1, b2, b3)) => (} \\
\hspace{0.7cm} \T{|$h_1$|(a1, b1)}, \\
\hspace{0.7cm} \T{|$h_2$|(a2, b2)}, \\
\hspace{0.7cm} \T{|$h_4$|(a3, b3))}
} \\ \hline

\end{tabular}
\end{table}

Our  approach  further simplifies the  synthesis problem through type-directed decomposition. In particular, rather than trying to synthesize the entire merge operator, we realize that each element in the output tuple can be synthesized independently, as shown in \Cref{fig:motivating-decomp}. In particular, each component of the original aggregation can be translated into a corresponding merge expression: the maximum bid is handled by $h_1$, the high-bid count by $h_2$, and the item bid count map by $h_4$. Furthermore, our approach further simplifies the normalizer synthesis problem for the item bid count map by synthesizing a normalizer $h_3$ for each individual entry in the map. In practice, such decomposition turns out to be crucial for handling real-world aggregations.

\vspace*{0.1in}
\begin{mdframed}
    \textbf{Idea 3:} We can combine inductive and deductive synthesis to make the solution more effective.
\end{mdframed}
\vspace*{0.1in}
As shown in \Cref{fig:motivating-decomp}, the synthesis problems for  $h_1, h_2, h_3$ are quite simple and involve only scalar operations. We refer to these as \emph{leaf-level synthesis problems} and  use standard  synthesis techniques based on  SyGuS  to solve them. However, our approach uses \emph{deductive reasoning} to both decompose the problem into independent subproblems and to combine their results. For instance, consider the normalizer $h_4$ from \Cref{fig:motivating-decomp}, which uses $h_3$ as a subexpression. Here, $h_3$ is synthesized using SyGuS, but, given a solution for $h_3$, our method can construct $h_4$ from $h_3$ using deductive synthesis, which obviates the need for searching over a large space of programs and crucially avoids unnecessary invocations of SyGuS for complex data structures such as maps. 

\vspace*{0.1in}
\begin{mdframed}
    \textbf{Idea 4:} We can refute the existence of merge operators without attempting synthesis.
\end{mdframed}
\vspace*{0.1in}

While this example admits a suitable merge operator, some programs are not homomorphisms and thus inherently lack a corresponding merge function. Instead of wasting  resources trying to synthesize a non-existent merge operator, our approach leverages  proof rules in the calculus to identify when a merge function cannot exist. Specifically, we establish criteria that detect that the accumulator’s behavior is incompatible with the existence of a merge function, which allows us to identify cases where the commutativity and identity conditions cannot be simultaneously satisfied.

%% file: figures/motivating-example.tex
\begin{figure}
\small
\centering
\begin{minted}
[
fontsize=\footnotesize,
escapeinside=||,
numbersep=5pt,
]
{scala}
case class BidData(bidPrice: Float, item: Int)
case class BidAggBuffer(maxBid: Float, highBidCount: Int, itemBidCounts: Map[Int, Int])

object BidAggregator extends Aggregator[BidData, BidAggBuffer, BidAggBuffer] {
  def zero: BidAggBuffer = BidAggBuffer(Float.MinValue, 0, Map.empty)

  def reduce(buffer: BidAggBuffer, data: BidData): BidAggBuffer = {
    val newMaxBid = math.max(buffer.maxBid, data.bidPrice)
    val newHighBidCount = if (data.bidPrice > 1000) buffer.highBidCount + 1 else buffer.highBidCount

    val itemBidCountMap = buffer.itemBidCountMap
    val newItemBidCounts = itemBidCountMap + (item -> (itemBidCountMap.getOrElse(item, 0) + 1))

    BidAggBuffer(newMaxBid, newHighBidCount, newItemBidCounts)
  } }

val result = bidsDF.filter(year(col("AuctionDate")) === 2024)
                   .select("BidPrice", "Item").as[BidData]
                   .agg(new BidAggregator()($"BidPrice", $"Item").as("aggregated_result"))
\end{minted} 
\vspace*{-0.1in}

\caption{A Scala \spark{} program used to compute auction information illustrating our motivating example. \todo{Here, \T{BidAggregator} is the UDAF, and the implementation of \T{reduce} is the corresponding accumulator function.}} 
\label{fig:motivating-example}
\end{figure}

%% file: problem.tex
\section{Problem Statement}\label{sec:prob_stmt}

In this paper, we consider a family of programs that perform aggregation over \emph{dataframes} through user-defined functions. 
In the rest of this section, we first define \emph{dataframes}, then introduce a domain-specific language (DSL) used in our formalization, and finally state our problem definition.

\begin{definition}\label{def:df}{\bf (Dataframe)}
A \emph{dataframe} $\df$ is a quadruple $(\dfcols, \dftype, \dfrows, \dfvals)$  where $\dfcols$ is a sequence  of column labels, $\dftype$ is a mapping from each $c_j \in \dfcols$ to its corresponding type $\tau_j$,   $\dfrows = [r_1, \ldots, r_n]$ is a list of rows, and $\dfvals: \dfrows \times \dfcols \rightarrow \tau$ is a mapping from each entry $(r_i, c_j)$ to a value $v \in \dftype(c_j)$. Given a dataframe $\df$, the \emph{type} of the dataframe is  $\mathsf{DF}\langle \tau_1, \ldots, \tau_n \rangle$ where $\tau_i = \dftype(c_i)$.
\end{definition}

\input{figures/dsl}

\Cref{fig:dsl-syntax} shows the syntax  of a DSL designed to express programs that perform aggregation over dataframes. At a high level, this DSL supports SQL-like queries incorporating user-defined functions (UDFs).
As shown in \Cref{fig:dsl-syntax}, the top-level program returns the result of an aggregation applied to the result of a \emph{data transformation program} $\trans$. A data transformation program transforms the input dataframe to a new dataframe using  standard relational operators such as {\tt select} and {\tt project}\footnote{Our implementation also supports {\tt groupBy}; however, we omit it here to simplify presentation.}, but these operators can also involve user-defined functions. The top-level program  $\lambda x. \agg(\udaf, \init, \dftrans)$ first  computes $\dftrans(x)$ to obtain a dataframe $\df$ and then applies the \emph{accumulator function} $\udaf: \tau_r \times \tau \rightarrow \tau_r$ to $\df$, using $\init$ as the initial value. 
Functions in this DSL are functional programs that  support the creation of complex data structures like maps, lists, sets (and nested combinations thereof) via higher-order operations like {\tt map} and {\tt fold}. Such user-defined functions are particularly useful in scenarios where data needs to be summarized into a structured form for further downstream analysis or processing.



\begin{example}

Consider a program that takes as input  a dataframe $\df : \dftypee{\TInt \times \TInt \times \TInt}$ and performs a frequency count of the second column of the input dataframe and stores them in a map of type $\TMap{\TInt}{\TInt}$. 
We can implement this program via our DSL as
\[
\mathcal{P} = \abslambda{t:\dftypee{\TInt \times \TInt \times \TInt}}{\agg(\udaf, \emptymap,\mathtt{project}(\sigma_2, t))}, \quad \mathrm{where}
\]
\[
f = \abslambda{s:\TMap{\TInt}{\TInt}}{\abslambda{x : \TInt}{\mathsf{ITE}(\mathsf{contains}(s, x), \mathsf{update}(s, x, \mathsf{get}(s, x) + 1), \mathsf{update}(s, x, 1))}}.
\]
Here, $\mathsf{contains}(s, x)$ and $\mathsf{get}(s, x)$ are built-in primitives for querying whether key $x$ is in map $s$ and retrieving the value of key $x$ in map $s$, respectively.
\end{example}

The problem that we address in this paper is to determine whether a program in this DSL corresponds to a \emph{dataframe homomorphism}. 
To precisely define our problem, we first introduce a concatenation operation $\dfconcat$ on dataframes as follows:

\begin{definition}\label{def:df-concat}{\textbf{(Dataframe concatenation)}} Let $\df_1 = (\dfcols, \mathcal{T}, \dfrows_1, \dfvals_1), \df_2 = (\dfcols, \mathcal{T}, \dfrows_2, \dfvals_2)$ be two dataframes. Then, $\df_1 \dfconcat \df_2 $ is defined as $(\dfcols, \mathcal{T}, \dfrows_1 \cup \dfrows_2, \dfvals)$ where:
\[ \dfvals(r_i, c_j) = \dfvals_1(r_i, c_j) \text{ if } i \leq |R_1|, \text{ else } \dfvals_2(r_{i - |R_1|}, c_j)\]

\end{definition}

For the purposes of this paper,  a \emph{dataframe aggregation} is any program that belongs to the DSL from \Cref{fig:dsl-syntax}. Using this terminology, we define \emph{dataframe homomorphism} as follows:

\begin{definition}[\textbf{Dataframe homomorphism}]
A program $\dfagg$ is a  homomorphism iff there \emph{exists} a function $\dfmerge: \tau \times \tau \rightarrow \tau$ such that, for any   dataframes $\df_1, \df_2$ on which $\dfconcat$ is defined, we have:
\begin{equation}\label{eq:df-homo}
\dfagg(\df_1 \dfconcat \df_2) = \dfagg(\df_1) \dfmerge \dfagg(\df_2)
\end{equation}
\end{definition}

Intuitively, if an aggregation $\dfagg$ is a homomorphism, we can partition a dataframe $\df$ into multiple dataframes $\df_1, \ldots, \df_n$, apply the aggregator $\dfagg$ to each $\df_i$ and then merge the results using the $\dfmerge$ operator.
In the rest of this paper, we refer to the binary operator $\dfmerge$ as the \emph{merge} function for the aggregator. 
Note that our definition does not require the merge function to be commutative. This design choice is deliberate, as it allows our framework to support applications like incremental computation, where partial results are naturally merged in a fixed, non-commutative order (e.g., merging an existing result with a result from new data).

 We conclude this section by defining the \emph{homomorphism verification} problem:

\begin{definition}[\textbf{Homomorphism verification problem}]
The homomorphism verification problem is to determine whether a program $\dfagg$ is a dataframe homomorphism, and, if so, construct a merge function $\dfmerge$ that satisfies \Cref{eq:df-homo}.
\end{definition}

%% file: figures/dsl.tex
\begin{figure}[t]
\[\arraycolsep=2pt\def\arraystretch{1.2}
\begin{array}{rlcl}
\textbf{Program} & \prog &::= &\abslambda{x: \type}{\agg(\udaf, \init, \dftrans)}  \\

\textbf{DataFrame} & \dftrans &::= &x ~|~ \proj(\udaf, \dftrans) ~|~\select(\udaf, \dftrans) \\

\textbf{Function} & \udaf &::= &\abslambda{x: \type}{\udaf} ~|~\abslambda{x: \type}{E} ~|~ g \\

\textbf{Expression} & E &::= &c ~|~ \default{\tau} ~|~  \ \udaf(E, \dots, E) ~|~ \ite(E, E, E) ~|~ \fold(\udaf, E, C) ~|~ (E, \dots, E) ~|~\typeselector_i(E) ~|~ C   \\

\textbf{Collection Expr} & C &::= &M ~|~ L ~|~ S ~|~ \convert_\type{(C)} ~|~ \map(\udaf, C) ~|~\filter(\udaf, C)  ~|~ \zip(C, C)  \\

\textbf{Map Expr} & M &::= &\emptymap ~|~ \update(M, E, E) ~|~M \collectionjoin M \\
\textbf{List Expr} & L &::= &\emptylist ~|~ \append(L, E) \\

\textbf{Set Expr} & S &::= &\emptymap ~|~ \union(S, S) ~|~ \setinsert(S, E)  \\
\end{array}
\]
\[
\begin{array}{c}
c \in \textbf{Constants} \quad x \in \textbf{Variables} \quad g \in \textbf{Built-in Functions} \\
\end{array}
\]

\caption{DSL syntax. The $\update$ function updates keys or adds new (key, value) pairs. $\typeselector_i$ returns the i'th tuple element, and $\default{\type}$ gives a default expression of type $\type$ (e.g., $0$ for Int). The $\collectionjoin$ operator performs an outer join of two maps $M_1$ and $M_2$, producing $M : \TMap{\tau_1}{\tau_2 \times \tau_2}$, where (1) $M(k) = (M_1(k), M_2(k))$ if $k$ is in both, (2) $M(k) = (\mathsf{null}, M_2(k))$ if only in $M_2$, and (3) $M(k) = (M_1(k), \mathsf{null})$ if only in $M_1$. $\convert_\tau(C)$ converts $C$ to type $\tau$: e.g., lists are converted to maps by using their indices as keys, and sets are converted to maps using each element as a key with a \textsf{null} value.}
        \label{fig:dsl-syntax}
\end{figure}

%% file: methodology.tex
\section{Homomorphism Calculus}
\label{sec:problem-statement}
This section presents a set of proof rules for reasoning about dataframe homomorphisms. Central to this calculus is the concept of a \emph{normalizer}, which serves as a bridge between the desired merge operator and the accumulator function  used inside the aggregation. 

\subsection{Foundation of the Calculus: Normalizers}
\label{ssec:normalizers}

As mentioned in Section~\ref{sec:intro},  synthesizing a merge operator for a dataframe  aggregation $\dfagg$ is challenging because it requires reasoning about  the behavior of $\dfagg$ on the \emph{entire} dataframe, which contains an unbounded number of rows. On the other hand, reasoning about the accumulator function $f$ is generally easier because it operates over a single row of the dataframe. In this section, we introduce the concept of normalizer  in order to bridge this complexity gap.  To formalize this concept, we first introduce a generalized notion of commutativity between actions on a set:

\begin{definition}[\textbf{Actions}]
Let $X, Y$ be two sets. 
A right action $\alpha_r$ of $X$ on a set $Y$ is a  function of type $Y \times X \rightarrow Y$, and a left action $\alpha_l$ of $X$ on a set $Y$ has signature $X \times Y \rightarrow Y$.
\end{definition}

\begin{wrapfigure}{h}{0.25\linewidth}
\vspace{-0.3in}
\begin{center}
\includegraphics[scale=0.3]{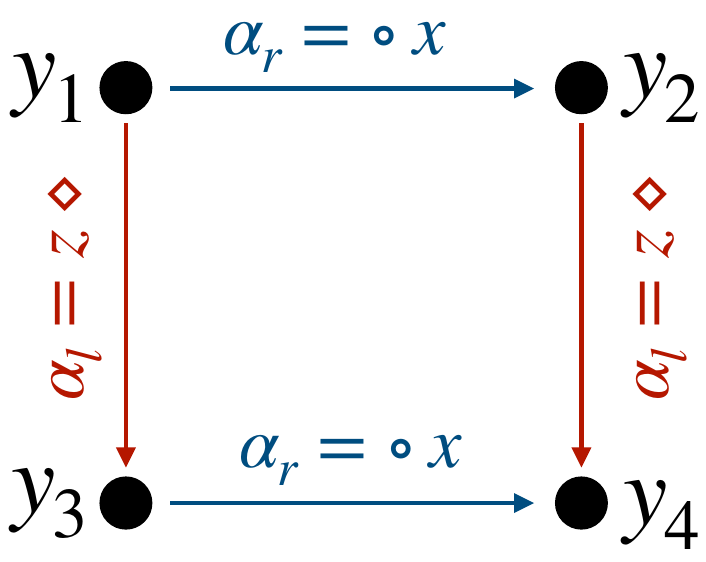}
\end{center}
\vspace{-0.2in}
\caption{Commutativity.}
\label{fig:commutativity}
\vspace{-0.3in}
\end{wrapfigure}

Intuitively, a right  action of $X$ ``hits'' elements of set $Y$ from the right  to produce another element of $Y$;  a left action does the same but from the left.  To relate this concept to our setting, consider an accumulator $\udaf$ of type $\type_r \times \type \rightarrow \type_r$, where $\type$ is the type of a single row of the dataframe and $\type_r$ is the type of the internal state of the accumulator. 
In our context, we can view $\udaf$ as a right action  of $\type$ on $\type_r$.


\begin{definition}[\textbf{Commutativity of  actions}]\label{def:commute} Let $\alpha_r$ be a right action of a set $X$ on a set $Y$, and let $\alpha_l$ be a left action of a set $Z$ on $Y$. Actions $\alpha_r, \alpha_l$ \emph{commute} iff:
\[
\forall x\in X, \forall y \in Y, \forall z \in Z. \ \  \alpha_l(z, \alpha_r(y, x))) = \alpha_r(\alpha_l(z, y), x) 
\]
\end{definition}
In other words, a right and left action on a set $Y$ commute with each other if the order in which we apply them does not matter. This is illustrated schematically in \Cref{fig:commutativity}. 

\begin{example}

Consider a function $ f = \abslambda{s}{\abslambda{x}{s + \mathsf{len}(x)}}  $ of type $ \TInt \to \TList{\TInt} \to \TInt$, which is a right-action of $\TList{\TInt}$ on $\TInt$.
Also, let $ g = \abslambda{x}{\abslambda{y}{\typeselector_1(x) + y}}$ be a function of type $ (\TInt \times \TInt) \to \TInt \to \TInt$, which is a left-action of $(\TInt \times \TInt)$ on $\TInt$.
These two functions commute because
 $g(z, f(y, x))  = \sigma_1(z) + y + \mathsf{len}(x)= f(g(z, y), x)$.
On the other hand, let $f' = \abslambda{s}{\abslambda{x}{s \cdot \mathsf{len}(x)}}$ be another right action of $\TList{\TInt}$ on $\TInt$. In this case, $g(z, f'(y,x)) = \sigma_1(z) + (y \cdot \mathsf{len}(x))$ whereas $f'(g(z, y), x) = (\sigma_1(z)+y)\cdot \mathsf{len}(x)$. Thus, $f'$ and $g$ do not commute.
\end{example}


Next, we define the concept of \emph{normalizer}\footnote{Our  use of the term \emph{normalizer} differs from its use in group theory, although it bears resemblances in some respects.} that plays a big role in our calculus:

\begin{definition}[\textbf{Normalizer}]\label{def:normalizer}
Let $\alpha$  be a right (resp. left) action of a set $X$ on $Y$. A normalizer of $\alpha$ is a left (resp. right) action $\beta$ of  $\boldsymbol{Y}$ {\bf on} $\boldsymbol{Y}$ such that $\alpha$ and $\beta$ commute according to \Cref{def:commute}.
\end{definition}

Intuitively, a normalizer of an action $\alpha$ on set $S$  is an action of $S$ on 
\emph{itself}  that commutes with $\alpha$. 

\begin{example}
Consider the  function $f = \abslambda{x:\TInt}{\abslambda{y:\TList{\TInt}}{x + \mathsf{len}(y)}}$ which is  a right action of $\TList{\TInt}$ on $\TInt$.
The function $h = \abslambda{x:\TInt}{\abslambda{y:\TInt}{x + y}}$ is a normalizer  for $f$.
\end{example}

In general,  normalizers are neither guaranteed to exist nor must be unique. 

\begin{example}\label{ex:norm-nonexist}
Consider the  function $f(y, z) = 0 \text{ if } y = z \text{ else } z$ which is a right action of $\mathbb{N}$ on $\mathbb{N}$.
\Cref{proof:ex-norm-nonexist} provides a proof that  a normalizer of $f$ does not exist.

\end{example}


 \begin{example}\label{ex:norm-nonunique}

 Consider the function $f : \TInt \to \TInt \to \TInt = \abslambda{s}{\abslambda{x}{s + 1}}$ which is a right-action of $\TInt$ on $\TInt$. Consider functions $h_1 : \TInt \to \TInt \to \TInt = \abslambda{s_1}{\abslambda{s_2}{s_2}}$ and $h_2 : \TInt \to \TInt \to \TInt = \abslambda{s_1}{\abslambda{s_2}{s_2 + c}}$ where $c$ is an arbitrary integer. In this case, both $h_1$ and $h_2$ are normalizers for $f$.

\end{example}

\subsection{From Normalizers to DataFrame Homomorphisms }

In this section, we relate normalizers to merge functions for dataframe homomorphisms. We start by stating the following theorem that underlies the soundness of our calculus:

\begin{theorem}\label{thm:norm-sound}
Let $\dfagg = \abslambda{x}{\agg(\udaf, \init, x)}$  be a program where $\udaf: \tau_r \times \tau \rightarrow \tau_r$ is a right action of $\tau$ on $\tau_r$, and let $h$ be a normalizer of $f$ satisfying $\forall s \in \tau_r. h(s, \init) = s$. Then, $\dfagg$ is a  homomorphism.
\end{theorem}

In other words, if we can find a normalizer of $f$ satisfying the  condition $\forall s \in \tau_r. h(s, \init) = s$,  we can guarantee that  $\dfagg$ is a homomorphism. Intuitively, this theorem is very useful because checking whether a function is a normalizer for $f$ is a simpler problem than checking whether a merge operator satisfies \Cref{eq:df-homo}. This is the case because the latter problem requires reasoning about the \emph{entire} dataframe, whereas the former requires reasoning about \emph{just one row} of the dataframe.  

 While the theorem helps \emph{prove} that a program is a dataframe homomorphism, a natural question is whether homomorphism verification can be \emph{fully} reduced to finding a normalizer for the accumulator. Recall from \Cref{ex:norm-nonexist} that normalizers are \emph{not} always guaranteed to exist, raising the question of whether a program $\dfagg$ can be a dataframe homomorphism without a normalizer for its accumulator $f$. We prove that this cannot occur if $\dfagg$ is a surjective function from $\dftypee{\tau}$ to $\tau_r$:

\begin{theorem}\label{thm:norm-complete}
Let $\dfagg = \abslambda{x}{\agg(\udaf, \init, x)}$ be a dataframe homomorphism where $\udaf$ is a right action of $\tau$ on $\tau_r$. Then, if $\dfagg$ is a surjective function from $\mathsf{DF}\langle \tau \rangle$ to set $ \tau_r$, a normalizer of $\udaf$ is guaranteed to exist.
\end{theorem}

According to this theorem,  we can also \emph{disprove} that $\dfagg$ is a homomorphism by showing that a normalizer for the accumulator $f$ \emph{does not exist} as long as $\dfagg$ is surjective (which is realistic for most practical use cases). 
Furthermore, even in cases where  $\dfagg$ is \emph{not} surjective, the completeness result can be generalized by changing the scope of the quantifiers in the normalizer definition to just values in the range of the aggregation function.

Finally, recall from \Cref{ex:norm-nonunique}, that normalizers may not be unique when they exist. This raises the question of whether there can be \emph{multiple} semantically different merge functions for a given dataframe homomorphism. This would be problematic because it would mean that we can construct multiple  merge operators that lead to different results. Fortunately, the following theorem states the uniqueness of normalizers under the side condition imposed by the initializer:

\begin{theorem}\label{thm:norm-unique}
Let $\dfagg = \abslambda{x}{\agg(\udaf, \init, x)}$  be a surjective dataframe homomorphism from $\mathsf{DF}\langle \tau \rangle$ to set $ \tau_r$ where $\udaf$ is a right action of $\tau$ on $\tau_r$.  There exists a unique normalizer $\norm$ of $\udaf$ satisfying  $\forall s \in \tau_r. \norm(s, \init) = s$.
\end{theorem}

\subsection{Calculus Overview}

Before going into the proof rules of our calculus, we first provide a high-level overview of its structure.  
Our calculus is comprised of three types of complementary proof rules:

\begin{enumerate}[leftmargin=*]
\item {\bf Homomorphism validation and refutation rules:} These top-level rules are used to validate or refute whether a given program is a dataframe homomorphism. If it is a homomorphism, these rules also produce the corresponding merge operator to prove that it is a homomorphism.

\item {\bf Normalizer validation and refutation rules:} These rules are employed to either construct a normalizer for the accumulator or prove that none can exist. The refutation rules provide necessary conditions for the non-existence of a normalizer, while the validation rules attempt to synthesize one using a combination of inductive and deductive synthesis techniques.

\item {\bf Type-directed decomposition:}
Some of the synthesis rules for normalizer construction rely on the  expression being in a specific syntactic form. The  goal of type-directed decomposition is to facilitate deductive synthesis by rewriting expressions to match  these syntactic forms.


\end{enumerate}

\subsection{Homomorphism Validation and Refutation Rules}
\label{ssec:horizontal-rules}
\input{figures/fig-horizontal}

\Cref{fig:horizontal} describes our first set of proof rules. 
Before explaining in detail, we first provide~intuition:
\\

\begin{mdframed}
{\bf Observation \#1:} Let $\dftrans$ be a data transformation expression with free variable $x$. Then, $\dftrans[(x_1 \dfconcat x_2)/x]$ can always be rewritten as $\dftrans[x_1/x] \dfconcat \dftrans[x_2/x]$.
\end{mdframed} 
\vspace{0.1in}

In our calculus, the rules labeled {\sc Top} and {\sc Agg} exploit this observation, and the rules  labeled {\sc Rel}, $\dfconcat$, and {\sc Var}  define how to transform $\dftrans[(x_1 \dfconcat x_2)/x]$\footnote{\todo{We use the standard notation $E[v/x]$ to denote the substitution of every free occurrence of variable $x$ in expression $E$ with $v$.}} into $\dftrans[x_1/x] \dfconcat \dftrans[x_2/x]$. \\

\begin{mdframed}
{\bf Observation \#2:} Let $E = \agg(\udaf, \init, \dftrans)$ and $f$ have type $ \tau_r \times \tau \rightarrow \tau_r$. If $\norm$ is a  normalizer of $f$ satisfying  $\forall s \in \tau_r. h(s, \init) = s$, then,  $E[(x_1 \dfconcat x_2)/x]$ and  $h(E[x_1/x], E[x_2/x])$ are equivalent.
\end{mdframed} 
\vspace{0.1in}

This observation follows from the previous one and \Cref{thm:norm-sound}. Building on this, the {\sc Agg} rule rewrites the aggregation  using the normalizer, while the {\sc Top} rule transforms $\dfbody[x_1\dfconcat x_2/x]$ into $\norm(\dfbody[x_1/x], \dfbody[x_2/x])$. Consequently, the merge operator for the function is defined as $\abslambda{r_1}{\abslambda{r_2}{h(r_1, r_2)}}$. Intuitively, the {\sc Top} rule demonstrates how to propagate the normalizer $\norm$ of the UDAF $f$ throughout the program, while some of the other rules, such as {\sc Rel}, justify its soundness. \\


\begin{mdframed}
{\bf Observation \#3:} Any dataframe homomorphism $\dfagg$ must satisfy the following axiom:
\[
\forall x_1, x_2, y_1, y_2. \ \dfagg(x_1) = \dfagg(x_2) \land \dfagg(y_1) = \dfagg(y_2) \rightarrow \dfagg(x_1 \dfconcat y_1) = \dfagg(x_2 \dfconcat y_2)
\]
\end{mdframed} 
\vspace{0.1in}
Intuitively, the above observation states that the semantics of $\dfagg$ must be consistent with the existence of a merge operator. To see why, recall the definition of homormorphism, which states that 
$\dfagg(\df_1) \dfmerge \dfagg(\df_2)  = \dfagg(\df_1 \dfconcat \df_2) 
$.
If we instantiate the function axioms for $\oplus$ in the above definition,  we obtain the formula from Observation \#3.  Hence, the negation of this observation provides a way to \emph{refute} that a program is a homomorphism, as formalized by the {\sc Refutation} rule in \Cref{fig:horizontal}. This rule states that, if we find two input pairs $(\df_1, \df'_1)$ and $(\df_2, \df'_2)$ where $\mathcal{P}$ produces the same output for each $(\df_i, \df_i')$ but produces a different output on $\df_1 \dfconcat \df_2$ vs. $\df'_1 \dfconcat \df'_2$, then $\mathcal{P}$ cannot be a homomorphism.


\begin{theorem}
\label{thm:soundness_completeness_hom_synthesis_rules}
A program $\prog$ is a homomorphism if and only if $\prog \horizdecompto \norm$ for some binary function $\norm$ according to the rules in \Cref{fig:horizontal}. Furthermore,  $\norm$ is the merge operator for  homomorphism $\prog$.
\end{theorem}

\subsection{Normalizer Construction and Refutation}
\Cref{fig:witness_construction} describes our proof rules for constructing a normalizer for the accumulator function or proving that none  exists. Recall from \Cref{thm:norm-sound,thm:norm-complete} that a program is a dataframe homomorphism if and only if  there exists a function $h$ such that (1) $h$ and $f$ commute (\Cref{def:commute}) and (2)  $\forall s. ~h(s, \init) = s$. An obvious strategy for constructing such a function $h$ is to use syntax-guided synthesis~\cite{sygus} by encoding the specification as a logical formula. However, because syntax-guided synthesis often requires searching over a large  space of programs, this approach does not work well in practice, particularly when the accumulator involves complex data structrues instead of scalar values.  The following observation underpins the design of our normalizer proof rules:

\vspace*{0.1in}
\begin{mdframed}
    {\bf Observation \#4:} Given an accumulator with complex internal state, we can often \emph{decompose} the normalizer synthesis problem to several simpler synthesis problems that only involve scalars. Furthermore, we can avoid performing search for unrealizable synthesis problems by leveraging necessary conditions for the existence of a normalizer.
\end{mdframed}
\vspace*{0.1in}

\input{./figures/fig-witness}

\Cref{fig:witness_construction} shows the normalizer proof rules in our calculus where $(f, \init) \sim h$ indicates that $h$ is the desired normalizer for accumulator $f$ with initializer $\init$. These rules can be grouped into three categories: The {\sc Norm-Synth} rule serves as the \emph{base case} and relies on an external SyGuS solver. The next two rules, labeled {\sc Norm-Product} and {\sc Norm-Collection}, decompose the normalizer synthesis problem for complex data types into  simpler problems involving less complex data types. Finally, the last two rules prove the non-existence of a normalizer. 
 
\bfpara{The \textsc{Norm-Synth} rule.} This rule leverages an external SyGuS solver (via the \textsf{Solve} procedure) to synthesize a normalizer. The generated SyGuS specification consists of two parts, with $\Phi_1$ encoding the initializer side condition   and $\Phi_2$ specifying  the commutativity condition (\Cref{def:commute}).\footnote{If a specific framework additionally requires the merge operator to be commutative, we additionally provide the constraint  $\forall a, b. h(a,b)=h(b,a)$ as part of  the synthesis query. However, as stated in \Cref{sec:prob_stmt}, our methodology also allows non-commutative merge operators for generality.}
In this rule, we assume that the output of $\mathsf{Solve}$ yields an implementation that satisfies the specification $\Phi_1 \land \Phi_2$.  In general, while SyGuS solvers are quite effective at solving synthesis problems that involve scalars, they empirically struggle with complex data structures. The next two rules aim to decompose such complex problem instances into a series of simpler, scalar-valued instances.

\bfpara{The \textsc{Norm-Product} rule.}
Since many real-world accumulators operate over tuples, the {\sc Norm-Product} rule  decomposes a function $\mathbf{f}$ whose return type  is a tuple into multiple sub-problems, each of which returns a single element of the tuple.  To do so, this rule first finds an expression of the form
$
(\mathbf{f_1}(\typeselector_1(a), b), \ldots, \mathbf{f_n}(\typeselector_n(a), b)) 
$ that is semantically equivalent to $\mathbf{f}$.
Then, instead of finding a single normalizer $\mathbf{h}$ for $\mathbf{f}$, this rule recursively synthesizes a separate normalizer $\mathbf{h_i} $ for each 
$\mathbf{f_i} $ and composes them via the tuple constructor.

\bfpara{The \textsc{Norm-Coll} rule.}
The next rule, labeled {\sc Norm-Coll}, applies to functions whose output is a collection of type $\typecollection{\tau}$ and simplifies the problem of constructing a normalizer with return type $\typecollection{\tau}$ to a simpler one with return type $\tau$. To do so, given a function $\mathbf{\udaf}$, it first finds a semantically equivalent expression of the form $\map(\abslambda{v}{\mathbf{\udaf'}(v, x)}, \filter(p, s)) $.
Intuitively, if  $\mathbf{\udaf}$  can be expressed in this  form,  we can reduce the problem of finding a normalizer for $\mathbf{\udaf}$ to the problem of finding a normalizer for $\mathbf{\udaf'}$: Since $\mathbf{\udaf'}$ applies a transformation to each element in the collection, the  merge function only needs to figure out how to combine each element pair-wise and then build the collection back up. Thus, the {\sc Norm-Coll} rule first finds a normalizer $\mathbf{\norm}$ for $\mathbf{\udaf'}$ and then constructs the desired  merge function  by applying $\mathbf{\norm}$ to each element pair-wise and combining the results. Note that this rule uses the outer join operation $\collectionjoin$ on maps from our DSL (see \Cref{fig:dsl-syntax}) and uses the $\convert$ operation (also from the DSL) to perform type conversion between different collection types.


\begin{example} \label{example:udaf_normalizer_synthesis}

Consider the following function $f: \tau_r \to \mathsf{Int} \to \tau_r$ where $\tau_r$ is a (Int, Map) pair:
\begin{align*}
    \udaf &= \abslambda{s}{\abslambda{x}{(\typeselector_1(s) + x, \map(\abslambda{(k, v)}{\ite(k = x, v + 1, v)}}, \typeselector_2(s)))}.
\end{align*}
The program $\abslambda{x}{\agg(\udaf, (0, \emptymap ), x)}$ takes as input a dataframe with one column of type integer and produces a tuple consisting of (1) the cumulative sum of all elements and (2) a frequency count of unique elements.
First, using the {\sc Norm-Product} rule, we decompose the normalizer synthesis problem for $f$ into two independent subproblems defined by the following functions: 
\begin{align*}
    \mathbf{f}_1 = \abslambda{s_1}{\abslambda{x}{s_1 + x}} \quad \quad 
    \mathbf{f}_2 = \abslambda{s_2}{\abslambda{x}{\map(\abslambda{(k, v)}{\ite(k = x, v + 1, v), s_2)}}}
\end{align*}
For $\mathbf{f_1}$, we use the {\sc Norm-Synth} rule to construct the  normalizer $ \mathbf{h}_1 = \abslambda{s_1}{\abslambda{s_2}{s_1 + s_2}}$. For $\mathbf{f}_2$, we further simplify it using the {\sc Norm-Coll} rule. In particular, we first realize that $\mathbf{f}_2$ can be written as:
\[ \abslambda{s_2}{\abslambda{x}{\map(\abslambda{(k, v)}{\ite(k = x, v + 1, v), \filter(\top, s_2))}}}.\]
Thus, according to  {\sc Norm-Coll}, we need to synthesize a normalizer for  
$
\mathbf{f}_3 = \abslambda{v}{\abslambda{x}{\ite(k = x, v + 1, v)}}
$ where the initializer is $\default{\mathsf{Int}} = 0$. Next, we again use the {\sc Norm-Synth} rule to construct the normalizer $\mathbf{h}_3 = \abslambda{v_1}{\abslambda{v_2}{v_1 + v_2}}$.
This gives the normalizer for $f_2$ as follows:
\[
\mathbf{h}_2 = \abslambda{s_1}{\abslambda{s_2}{\map(\abslambda{(k, v_1, v_2)}{(k, v_1 + v_2)}, s_1 \collectionjoin s_2)}}.
\]
Finally, we obtain the merge function $\dfmerge$ for the program as
\[
\abslambda{s_1}{\abslambda{s_2}{(\typeselector_1(s_1) + \typeselector_1(s_2),  \map(\abslambda{(k, v_1, v_2)}{(k, v_1 + v_2)}, \typeselector_2(s_1) \collectionjoin \typeselector_2(s_2)))}}.
\]

\end{example}

\bfpara{Refutation rules.} As discussed earlier, some accumulators may not have a corresponding normalizer. In such cases, we would like to prove  unrealizability instead of searching for a solution that is guaranteed not to exist. To address this issue, our calculus contains two rules for disproving the existence of a normalizer. These rules are based on the following theorem:


\begin{theorem}
\label{thm:unrealizability}
    Given a function  $\udaf : \type_r \times \type \to \type_r$ and initializer $\init$ of type $ \type_r$, the following are necessary conditions for the existence of a suitable normalizer for  $f$:
    \begin{enumerate}
        \item $\forall s : \tau_r. \forall x : \tau.  \ (\udaf(\init, x) = \init \implies \udaf(s, x) = s$).
        \item $\forall s : \tau_r. \forall x, x' : \tau. \ (\udaf(\init, x) = \udaf(\init, x') \implies \udaf(s, x) = \udaf(s, x'))$.
    \end{enumerate}

\end{theorem}
\begin{example}

Consider $f = \lambda s. \lambda x. \textsf{ITE}(s = \init, \init, \mathsf{append}(1, s))$ of type $ \TList{\TInt} \times \TInt \to \TInt$ where $\init$ is the empty list. The existence  a normalizer can be refuted by rule \textsc{Norm-Refute-1} because $f(\init, x) = \init$ for all $x$, but $f(s, x) \neq s$ for any $s$. Similarly, we can refute the existence of a normalizer 
for 
$f = \lambda s. \lambda x. \mathsf{ITE}(s = \init, [1], \mathsf{append}(x, s))$, 
 using  the {\sc Norm-Refute-2} rule:   $f(\init, 3) = f(\init, 0) = \init$, but for $s=[1, 2]$, we have $f([1, 2], 3) = [1, 2, 3] \neq f([1, 2] , 0) = [1, 2, 0]$. 

\end{example}

{We conclude this section with the following theorem that states the soundness and completeness of the normalizer synthesis rules:}

\begin{theorem}\label{thm:norm-rules-soundness}
If $(f, \init) \sim h$, then $h$ is a normalizer of $f$ satisfying $\forall s.~~h(s, \init) = s$ iff $h \neq \bot$.
\end{theorem}

According to this theorem, our proof rules are sound, meaning that the resulting function $h$ is guaranteed to be a normalizer of $f$ satisfying the initial condition. Furthermore, if the proof rules  refute the existence of a normalizer (meaning that they produce $\bot$), then a normalizer of $f$ satisfying the initial condition does not exist.






\subsection{Expression Decomposition}
\label{ssec:expr_decompose}
Recall that some of the rules in \Cref{fig:witness_construction} require rewriting  function $f$ into a specific \emph{syntactic} form. The last set of proof rules in our calculus describes how to do so based on the following observation:

 \vspace*{0.1in}
\begin{mdframed}
    {\bf Observation \#5:} Given a function $f$ with parameter $x$ of type $\tau$, we can often convert $f$ to a semantically equivalent function $f'$ whose argument has a simpler type $\tau'$.
\end{mdframed}
\vspace*{0.1in}

 


 

\begin{table}[t]
\small
\centering
\caption{Syntax of decomposed expressions.}
\vspace{-0.1in}
\begin{tabular}{|l|l|}
\hline
\textbf{Form} & \textbf{Description} \\ \hline
$\isf_{\phantom{T}} = \vdcomp{\abslambda{x}{E}}{\typedestructor} ~|~ \vdcomp{\abslambda{x}{\isf}}{\typedestructor}$ &
Function \\ \hline

$\dtuple = \bounded{\vdfunc}{\decomp_1, \ldots, \decomp_n}$ &
Tuple \\ \hline

$\dcollection = \unbounded{\typeiter}{\decomp}{\vdpred}{\typedestructor} ~|~ \unboundedapp{\typeiter}{\decomp}{\vdpred}{x \in X}$ & 
Collection \\ \hline

$\decomp_{\phantom{T}} = E \mid \isf \mid \dtuple \mid \dcollection$ & 
Expression \\ \hline

\end{tabular}
\label{table:decomposed}
\end{table}
To gain intuition about why this is the case, consider a function $f: (\tau_1 \times \ldots \times \tau_n) \to \tau$ that takes as input a tuple but only touches one   element of the tuple. In this case, we can obtain a semantically equivalent function $f'$ of type $\tau_i \rightarrow \tau$. Our calculus makes use of  this observation through \emph{decomposed} expressions, whose syntax is provided in 
\Cref{table:decomposed}.  As shown here, there are four types of  \emph{decomposed expressions}: (1) standard expressions $E$ that cannot be further decomposed, (2) decomposed functions $\isf$, (3) decomposed tuples $\dtuple$, and (4)  decomposed collections $\dcollection$. Importantly, a decomposed function is of the form $\vdcomp{\abslambda{x}{E}}{d}$ where the auxiliary function $d$ is used for ``destructing'' the input into something simpler. Intuitively, if our calculus rewrites a lambda abstraction $\abslambda{x:\type}{E}$ into  $\abslambda{x':\type}{E'} \circ d$, this means:
\[
(\abslambda{x:\type}{E}) \ (E_0) \equiv (\abslambda{x':\type'}{E'}) \ (d(E_0))
\]
Crucially, because the input type $\tau'$ of the decomposed abstraction is simpler than than of the original type $\tau$, our calculus allows gradually simplifying functions that take complex inputs into functions that take simpler inputs (e.g., an integer instead of a list of integers). 

Next, we give a high level overview of the expression decomposition rules in \Cref{fig:vertical_decomp}. These rules describe how to convert an expression $E$ into  a decomposed expression $\decomp$, which can then  be converted back into standard  expressions in our DSL using the rules shown in \Cref{fig:convert}.

\input{figures/fig-decomp-meaning}
 \vspace*{0.1in}
\begin{mdframed}
    {\bf Observation \#6:} After converting decomposed expressions into standard expressions using the rules shown in \Cref{fig:convert}, the resulting expressions match the syntactic forms required by the {\sc Norm-Product}  and {\sc Norm-Coll} rules used for normalizer synthesis. 
\end{mdframed}
\vspace*{0.1in}

Based on this observation, our method first converts a given expression to its equivalent decomposed form (using the rules in \Cref{fig:vertical_decomp,fig:pp}) and then obtains a standard expression of the required syntactic form through the conversion rules in \Cref{fig:convert}.

\bfpara{Semantics of decomposed expressions.} \Cref{fig:convert} defines the semantics of decomposed expressions by showing how they can be converted to standard expressions using judgments of the form $\decomp \decomptoexpr E $ where $\decomp$ is a decomposed expression and $E$ is a standard expression. As motivated earlier, the {\sc Abs} rules convert a decomposed abstraction $\abslambda{x}{E \circ d}$ into a standard expression as $\abslambda{x}{E[d(x)/x]}$. The {\sc Tuple} rule recursively converts the nested decomposed expressions $\decomp_i$ to standard expressions $E_i$ and constructs a new tuple (or a lambda abstraction that returns a tuple, depending on whether $E_i$'s are abstractions or not). Finally, the collection rules {\sc C1} and {\sc C2} translate decomposed collections to standard expressions that involve \texttt{map} and \texttt{filter}. The only difference between the {\sc C1} and {\sc C2} rules is whether $Y$ is a free or bound variable in the resulting expression.

\input{./figures/fig-vertical}
 
\bfpara{Decomposition rules.} Next, we turn our attention to  \Cref{fig:vertical_decomp} for converting standard expressions to decomposed expressions. These rules use judgments of the form $E \vertdecompto \decomp$, indicating that $\decomp$ is the decomposed version of $E$. At a high level, all of these rules recursively decompose any nested expressions in the premises and build back up a new decomposed expression. Since the {\sc BaseType} and {\sc Tuple} rules are self-explanatory, we only explain the remaining rules.
The rules labeled {\sc Collection, Map, Filter} generate decomposed collections of the form $\unboundedapp{\typeiter}{\decomp}{\vdpred}{x \in X}$. Such a decomposed collection expression represents iterating over all elements $x \in X$ that satisfy predicate $\vdpred$ such that each element $x$ is transformed using decomposed expression  $\decomp$. The {\sc Collection} rule applies to variables $X$ of type collection, which are represented using the decomposed expression $\unboundedapp{\typeiter}{x}{\top}{x \in X}$. The {\sc Filter} rule first recursively decomposes the nested expression $e$ as $\unboundedapp{\typeiter}{\decomp}{\phi}{x \in X}$ and then conjoins the filter predicate $p(e_v)$ with $\phi$ to obtain the new expression. The {\sc Map} rule is similar, but it also applies $f$ to decomposed expression $\decomp$ to obtain a new $\decomp'$. We illustrate the collection decomposition rules through the following example:

\begin{example}
\newcommand{\inc}{\mathsf{inc}}
\newcommand{\double}{\mathsf{double}}
\newcommand{\predgtzero}{x > 0}
Consider the   expression
$
\map(\mathsf{\double}, \filter(\abslambda{x}{\predgtzero}, \map(\inc, X)))$
where \textsf{inc}, \textsf{double} are built-in functions. To decompose this expression, we first rewrite $X$ as $\unboundedapp{\typeiter}{x}{\top}{x \in X}$. Applying the {\sc Map} rule, we then obtain $\unboundedapp{\typeiter}{\inc(x)}{\top}{x \in X}$. Then, applying the {\sc Filter} rule, we obtain  $\unboundedapp{\typeiter}{\inc(x)}{\inc(x) > 0}{x \in X}$. A final application of {\sc Map} yields $ \unboundedapp{\typeiter}{\double(\inc(x))}{\inc(x) > 0}{x \in X}$.

\end{example}
Next, we consider the rules for lambda abstractions. The {\sc Lam-Base} rule handles built-in functions and is straightforward. The {\sc Lam-Ind} rule converts the body $e$ into a decomposed expression $\decomp$, then derives a new decomposed abstraction $\decomp'$ using the rules in \Cref{fig:pp}. 


The function simplification rules in \Cref{fig:pp} use Observation \#5 to adjust the input types of lambda abstractions within $\decomp$, employing judgments of the form $(x:\tau, d, \decomp) \propagateparamto \decomp'$. Here, the triple $(x:\tau, d, \decomp)$ represents a decomposed abstraction that takes an argument $x:\type$ and computes $\decomp[d(x)/x]$. The resulting expression $\decomp'$ is semantically equivalent but applies Observation \#5 to identify simplification opportunities. For example, the {\sc Tuple-Inductive} rule in \Cref{fig:pp} modifies the input types of nested expressions when only a specific component of the input is used.


\input{./figures/fig-pp}


\begin{example}

Consider again the function from \Cref{example:udaf_normalizer_synthesis}:
\begin{align*}
    \udaf &= \abslambda{s}{\abslambda{x}{(\typeselector_1(s) + x, \map(\abslambda{(k, v)}{\ite(k = x, v + 1, v})}, \typeselector_2(s)))}.
\end{align*}
First, we can use the \textsc{Lam-Ind} rule to decompose the abstraction body. In this case, since the UDAF body is a tuple, we use the \textsc{Tuple} rule, which creates a decomposed tuple with two sub-expressions,  $\typeselector_1(s) + x$ (from the \textsc{BaseType} rule) and $\unboundedapp{\typeiter}{\ite(k = x, v + 1, v)}{\top}{(k, v) \in \typeselector_2(s)}$ (from the \textsc{Collection} and \textsc{Map} rule).
Once the body is decomposed, the \textsc{Lam-Ind} rule processes the inner abstraction with parameter $x$ by using the function simplification rules  (\Cref{fig:pp}), which results in:
\begin{align*}
    \bounded{g}{\abslambda{x}{\typeselector_1(s) + x}, \unboundedapp{\typeiter}{\abslambda{x}{\ite(k = x, v + 1, v)}}{\top}{(k, v) \in \typeselector_2(s)}}
\end{align*}
Next, the \textsc{Lam-Ind} rule processes the top-level abstraction with parameter $s$, which can be further decomposed via \textsc{Tuple-Inductive}: This rule replaces tuple accesses $\typeselector_i$ with fresh variable $v_i$ and recursively invokes the function simplification rules, yielding:
\[
(v_1: \TInt, \typeselector_1, \abslambda{x}{v_1 + x}) ~~\textbf{and}~~
(v_2: \TList{\TInt}, \typeselector_2, \unboundedapp{\typeiter}{\abslambda{x}{\ite(k = x, v + 1, v)}}{\top}{(k, v) \in v_2}).
\]
The first function is immediately simplified to $\vdcomp{\abslambda{v_1}{(\abslambda{x}{v_1 + x})}}{\typeselector_1}$ by the \textsc{Function} rule. For the second function, since $v_2$ is a collection type, we use the \textsc{C-Ind} rule that replaces the free variable in the original decomposed collection with $v_2$, which yields
\[
\unbounded{\typeiter}{\abslambda{v}{\abslambda{x}{\ite(k = x, v + 1, v)}}}{\top}{\typeselector_2}.
\]
Putting these decomposed expressions together, we have the complete decomposed expression for $\udaf$, where every sub-function operates on a simpler type than the original function does:
\begin{align*}
    \bounded{g}{\vdcomp{\abslambda{v_1}{(\abslambda{x}{v_1 + x})}}{\typeselector_1}, \unbounded{\typeiter}{\abslambda{v}{\abslambda{x}{\ite(k = x, v + 1, v)}}}{\top}{\typeselector_2}}.
\end{align*}
Finally, by using the conversion rules from \Cref{fig:convert}, we obtain the following simplified function:
\begin{align*}
\udaf &= \abslambda{(s: (\TInt, \TMap{\TInt}{\TInt}))}{\abslambda{(x: \TInt)}{(\mathbf{f}_1(\typeselector_1(s), x), \mathbf{f}_2(\typeselector_2(s), x))}}, ~~\textbf{where:} \\
\mathbf{f}_1 &= \abslambda{(v_1: \TInt)}{\abslambda{(x: \TInt)}{v_1 + x}} \quad  \quad 
\mathbf{f}_3 = \abslambda{(v: \TInt)}{\abslambda{(x: \TInt)}{\ite(k = x, v + 1, v)}}
\\
\mathbf{f}_2 &= \abslambda{(v_2: \TMap{\TInt}{\TInt})}{\abslambda{(x: \TInt)}{\map(\abslambda{(k, v)}{\mathbf{f}_3(v, x)}, \filter(\top, v_2))}}
\end{align*}
Note that the functions $\mathbf{f_1}, \mathbf{f_3}$ correspond to the functions for which we need to synthesize normalizers. Hence, by using our decomposition rules, we have reduced the problem of synthesizing a normalizer for the complex type $(\TInt, \TMap{\TInt}{\TInt})$ to the problem of synthesizing two normalizers for the much simpler type $\TInt$ (i.e., internal state of the functions $\mathbf{f_1}, \mathbf{f_3}$).

\end{example}

\begin{theorem}
    \label{thm:sound-simplification}
 If   $E \vertdecompto \decomp$ and $\decomp \decomptoexpr E'$, then $E$ and $E'$ are semantically equivalent.
 
\end{theorem}

\subsection{Putting it all Together: End-to-End Algorithm}

\input{figures/pseudocode}

We conclude this section by formulating  a verification  algorithm, summarized in  \Cref{alg:top-level}, based on our calculus. This procedure takes as input an aggregation program $\dfagg = \abslambda{x}{\agg(f, \init, x)}$ and either returns $\bot$  or a valid merge operator proving that $\dfagg$ is a homomorphism.
Starting at \cref{lst:top-level:step1}, the algorithm first tries to apply the homomorphism refutation rules in \Cref{fig:horizontal} to refute the given problem instance. If refutation fails, it  proceeds to synthesize a normalizer. To this end,  it first tries to decompose $f$ using the  rules from \Cref{ssec:expr_decompose}.  If $f$ cannot be decomposed, \cref{lst:top-level:step2} attempts to prove the existence of a normalizer for $f$ using  the {\sc Norm-Synth}  rule, which invokes a SyGuS solver.   On the other hand, if $f$ is decomposable, we use the technique from \Cref{ssec:expr_decompose} to first obtain a decomposed expression $\decomp$, which is then converted back to a standard expression $E$ (using the rules from \Cref{fig:convert}) such that $E$ is in a syntactically canonical form. Hence, when we pattern match against $E$ in the {\sc Norm-Coll} and {\sc Norm-Tuple} rules of \Cref{fig:witness_construction}, we can identify all the sub-functions $F$ for which we need to synthesize normalizers. This set $F$ is computed via the call to procedure {\sc MatchNormInductive} at \cref{lst:top-level:step3} of the algorithm. Then, the loop in lines \ref{lst:top-level:step3}-\ref{lst:top-level:step3-end} computes a normalizer for each $f_i \in F$ and adds the pair $(f_i, h_i)$ to a set $S$. If normalizer computation fails for any $f_i$, the procedure falls back on the {\sc ApplyNormSynth} rule. Otherwise,  the call to {\sc ApplyNormInductive} at \cref{lst:top-level:step5} computes a normalizer for the whole function by using the {\sc Norm-Coll} and {\sc Norm-Tuple} rules.

An important point about this algorithm is that it does \emph{not} return $\bot$ if the call to {\sc ApplyNormRefute} returns true at \cref{lst:top-level:norm-refute}. {This is due to the fact that the {\sc Norm-Coll} and {\sc Norm-Tuple} rules can be used to \emph{prove} the existence of a normalizer but \emph{not} for refuting it.} 
{In other words, there can be cases where the accumulator function is decomposable, but its corresponding normalizer is not.}
The following  example illustrates such a case:

\begin{example}\label{ex:incomplete-decomp}
Consider the  function $f: \abslambda{s:\mathsf{Bool} \times \mathsf{Int}}{\abslambda{x:\mathsf{Int}}{(\vbooltrue, x)}}$ with initializer $\init = (\vboolfalse, 0)$. Using the {\sc Norm-Product} rule, we decompose the normalizer synthesis problem for $f$ into two independent subproblems defined by the following functions: 
\begin{align*}
    \mathbf{f}_1 = \abslambda{s_1}{\abslambda{x}{\vbooltrue}} \quad \quad 
    \mathbf{f}_2 = \abslambda{s_2}{\abslambda{x}{x}}.
\end{align*}
Note that $\mathbf{f}_2$ does not have a normalizer, as can be confirmed using {\sc Norm-Refute-1}. However, the original function $f$ does have a normalizer, namely $\mathbf{h} = \abslambda{s_1}{\abslambda{s_2}{\ite(\typeselector_1(s_2), s_2, s_1)}}.$
\end{example}

\begin{theorem}
\label{thm:completeness_of_decomposition}
Let $h$ be the return value of \textsf{IsHomomorphism}($\dfagg$) where $\dfagg$ is a surjective function from $\mathsf{DF}\langle \tau \rangle$ to set $ \tau_r$. Given a sound and complete oracle \textsf{Solve}  for syntax-guided synthesis, and assuming the entire DSL in \Cref{fig:dsl-syntax} is provided to the SyGuS solver, we have $h \neq \bot$  if and only if $\dfagg$ is a dataframe homomorphism. Furthermore, if $h \neq \bot$, then $h$ corresponds to the binary operator that proves that $\dfagg$ is a homomorphism.
\end{theorem}



%% file: figures/fig-horizontal.tex
\begin{figure}
\vspace{-0.1in}
\small
\small
\[
\begin{array}{c}

\irulelabel
{\begin{array}{c}
\dfbody[(x_1 \dfconcat x_2) / x] \horizdecompto \norm(\dfbody[x_1/x], \dfbody[x_2/x])
\end{array}}
{\prog = \abslambda{x}{\dfbody} \horizdecompto \abslambda{r_1}{\abslambda{r_2}{h(r_1, r_2)}}}
{\textsc{(Top)}}

\\ \ \\

\irulelabel
{\begin{array}{c}
\udaf: \type_r \times \type \rightarrow \type_r \quad
\normof{\type_r}{\udaf}{\init} = \norm \quad
\dftrans \horizdecompto \dftrans_1 \dfconcat \dftrans_2
\end{array}}
{\agg(\udaf, \init, \dftrans) \horizdecompto h\left(\agg(\udaf, \init, \dftrans_1), \agg(\udaf, \init, \dftrans_2)\right)}
{\textsc{(Agg)}}

\\ \ \\

\irulelabel
{\begin{array}{c}
\alpha \in \set{\project, \select} \quad
\dftrans \horizdecompto \dftrans_1 \dfconcat \dftrans_2 \\
\end{array}}
{\alpha(\udaf, \dftrans) \horizdecompto \alpha(\udaf, \dftrans_1) \dfconcat \alpha(\udaf, \dftrans_2)}
{\textsc{(Rel)}}



\quad



\irulelabel
{\begin{array}{c}
\textsf{IsVar}(x)
\end{array}}
{x \horizdecompto x}
{\textsc{(Var)}}

\\ \ \\

\irulelabel
{\begin{array}{c}
\prog(\df_1) = \prog(\df_1') \quad
\prog(\df_2) = \prog(\df_2') \quad
\prog(\df_1\dfconcat{}\df_2) \neq \prog(\df_1'\dfconcat{}\df_2')
\end{array}}
{\prog \horizdecompto \bot}
{\textsc{(Refutation)}}

\end{array}
\]
\vspace{-0.2in}
\caption{Rules for homomorphism validation and refutation.
}
\label{fig:horizontal}
\vspace{-0.15in}
\end{figure}

%% file: figures/fig-witness.tex
\begin{figure}
\scriptsize
\small
\[
\begin{array}{c}

\irulelabel
{\begin{array}{c}
\udaf : \type_r \times \type \rightarrow \type_r \quad 
\Phi_1 \equiv \forall (r: \type_r). ~h(r, \init) = r \\
\Phi_2 \equiv \forall (a, b: \type_r), \forall (x: \type). ~h(a, \udaf(b, x)) = \udaf(h(a, b), x)
\end{array}}
{\normalizes{\udaf}{\init}{\textsf{Solve}(\Phi_1 \land \Phi_2)}}
{\textsc{(Norm-Synth)}}

\\ \ \\

\irulelabel
{\begin{array}{c}
\mathbf{f} : (\tau_1, \ldots, \tau_k) \times \tau \to   (\tau_1, \ldots, \tau_k) \\

\mathbf{\udaf}(s, x) \triangleq 
(\mathbf{f_1}(\typeselector_1(s), x), \ldots, \mathbf{f_n}(\typeselector_n(s), x)) \quad \mathsf{where} \quad 
\normalizes{\mathbf{f_i}}{\sigma_i(\init)}{\mathbf{h_i}} \\

\end{array}}
{
\normalizes{\mathbf{\udaf}}{\init}{\abslambda{(s_1, s_2)}{ 
\left( \mathbf{h_1}(\typeselector_1(s_1), \typeselector_1(s_2)), \ldots, \mathbf{h_n}(\typeselector_n(s_1), \typeselector_n(s_2))\right)}}
}
{\textsc{(Norm-Tuple)}}

\\ \ \\

\irulelabel
{\begin{array}{c}
\mathbf{\udaf}(s, x) \triangleq \map(\abslambda{v}{\mathbf{\udaf'}(v, x)}, \filter(p, s)) \quad
\normalizes{\mathbf{\udaf'}}{\default{\type}}{\mathbf{h}}
\end{array}}
{
\normalizes{\mathbf{\udaf}}{\default{\typecollection{\type}}}{\abslambda{(s_1, s_2)}{
\convert_{\typecollection{\type}}(\set{~(k, \mathbf{h}(v_1, v_2)) ~\mid~ (k, v_1, v_2) \in \convert_\textsf{Map}(s_1) \collectionjoin \convert_\textsf{Map}(s_2)~})
}}}
{\textsc{(Norm-Coll)}}

\\ \ \\

\irulelabel
{\begin{array}{c}
\exists x, s. ~~f(\init, x) = \init \land f(s, x) \neq s
\end{array}}
{\normalizes{f}{\init}{\bot}}
{\textsc{(Norm-Refute-1)}}

\\ \ \\

\irulelabel
{\begin{array}{c}
\exists s. \exists x, x'. ~~f(\init, x) = f(\init, x') \land f(s, x) \neq f(s, x')
\end{array}}
{\normalizes{f}{\init}{\bot}}
{\textsc{(Norm-Refute-2)}}

\end{array}
\]
\vspace{-0.15in}
\caption{Rules for normalizer validation and refutation. In the \textsc{Norm-Coll} rule, $F_\mathsf{Map}(s)$ converts an arbitrary collection into a map, and $\boxtimes$ corresponds to the outer join operator for maps from our DSL.}
\label{fig:witness_construction}
\vspace{-0.15in}
\end{figure}

%% file: figures/fig-decomp-meaning.tex
\begin{figure}[t]
\vspace{-0.2in}
\newcommand{\args}{\bar{x}}
\newcommand{\filterSmall}{\textsf{flt}}

\footnotesize
\[
\begin{array}{c}

\irulelabel
{\phantom{\isf \decomptoexpr E}}
{
\vdcomp{\abslambda{x}{E}}{d} \decomptoexpr \abslambda{x}{E[d(x)/x]}
}
{\textsc{(Abs-Base)}}

\quad

\irulelabel
{\isf \decomptoexpr E}
{\vdcomp{\abslambda{x}{\isf}}{d} \decomptoexpr \abslambda{x}{E[d(x)/x]}}
{\textsc{(Abs-Ind)}}

\\ \ \\

\irulelabel
{\phantom{\begin{array}{c}
E \decomptoexpr E
\end{array}}}
{E \decomptoexpr E}
{\textsc{(Expr)}}

\quad

\irulelabel
{\begin{array}{c}
\decomp_1 \decomptoexpr E_1 \quad \ldots \quad \decomp_n \decomptoexpr E_n \quad
\args \equiv \bigcup_i \params(E_i)
\end{array}}
{\bounded{G}{\decomp_1, \ldots, \decomp_n} \decomptoexpr
\abslambda{\args}{\left( E_1(\args), \ldots, E_n(\args)\right)}
}
{\textsc{(Tuple)}}

\\ \ \\

\irulelabel
{\begin{array}{c}
\decomp \decomptoexpr E \quad \args \equiv \params(\decomp)
\end{array}}
{\unboundedapp{\typeiter}{\decomp}{\vdpred}{y \in Y} \decomptoexpr \abslambda{\args}{\map(E(\args), \filterSmall(\abslambda{y}{\vdpred}, Y))}}
{\textsc{(C1)}}

\irulelabel
{\begin{array}{c}
\decomp \decomptoexpr E \quad \args \equiv \params(\decomp) \setminus \set{y}
\end{array}}
{\unbounded{\typeiter}{\decomp}{\vdpred}{\typedestructor} \decomptoexpr \abslambda{Y}{\abslambda{\args}{\map(E(y, \args), \filterSmall(\vdpred, \typedestructor(Y)))}}}
{\textsc{(C2)}}

\end{array}
\]
\vspace{-0.1in}
\caption{Semantics of decomposed expressions. \filterSmall \ stands for filter.}
\label{fig:convert}
\vspace{-0.1in}
\end{figure}

%% file: figures/fig-vertical.tex
\begin{figure}
\footnotesize
\[
\begin{array}{c}

\irulelabel
{\begin{array}{c}
e = (e_1, \ldots, e_n) \quad
e_1 \vertdecompto \decomp_1 \quad \ldots \quad e_n \vertdecompto \decomp_n
\end{array}}
{(e: \typeprod) \vertdecompto \bounded{\vdfunc}{\decomp_1, \ldots, \decomp_n}}
{\textsc{(Tuple)}}

\quad

\irulelabel
{\begin{array}{c}
\textsf{IsIdentifier}(X) \quad \fresh(x: \type)
\end{array}}
{(X: \typecollection{\type}) \vertdecompto \unboundedapp{\typeiter}{x}{\predtrue}{x \in X}}
{\textsc{(Collection)}}

\\ \ \\

\irulelabel
{\begin{array}{c}
\phantom{e \vertdecompto \decomp}
\end{array}}
{(e: \typebase) \vertdecompto e}
{\textsc{(BaseType)}}

\quad

\irulelabel
{\begin{array}{c}
\udaf \in \mathsf{BuiltIn}
\end{array}}
{\udaf \vertdecompto \udaf}
{\textsc{(Lam-Base)}}

\quad

\irulelabel
{\begin{array}{c}
e \vertdecompto \decomp
\quad 
(x: \type, \vdid, \decomp) \propagateparamto \decomp'
\end{array}}
{\abslambda{(x: \type)}{e} \vertdecompto \decomp'}
{\textsc{(Lam-Ind)}}

\\ \ \\

\irulelabel
{\begin{array}{c}
e \vertdecompto \unboundedapp{\typeiter}{\decomp}{\vdpred}{x \in X}
\quad
\decomp \decomptoexpr e_v \quad
f(e_v) \vertdecompto \decomp'
\end{array}}
{\map(f, e) \vertdecompto \unboundedapp{\footnotesize\typeiter}{\decomp'}{\vdpred}{x \in X}}
{\textsc{(Map)}}

\quad

\irulelabel
{\begin{array}{c}
e \vertdecompto \unboundedapp{\typeiter}{\decomp}{\vdpred}{x \in X} \quad
\decomp \decomptoexpr e_v
\end{array}}
{\filter(p, e) \vertdecompto \unboundedapp{\footnotesize\typeiter}{\decomp}{\vdpred \land p(e_v)}{x \in X}}
{\textsc{(Filter)}}

\end{array}
\]
\vspace{-0.1in}
\caption{UDAF decomposition rules (shown for a core subset of the expressions).}
\label{fig:vertical_decomp}
\end{figure}

%% file: figures/fig-pp.tex
\begin{figure}
\footnotesize
\[
\begin{array}{c}

\irulelabel
{\begin{array}{c}
\end{array}}
{(x: \type, \typedestructor, E) \propagateparamto \vdcomp{\abslambda{(x: \type)}{E}}{\typedestructor}}
{\textsc{(Expr)}}

\quad

\irulelabel
{\begin{array}{c}
\end{array}}
{(x: \type, \typedestructor, \isf) \propagateparamto \vdcomp{\abslambda{(x: \type)}{\isf}}{\typedestructor}}
{\textsc{(Function)}}

\\ \ \\

\irulelabel
{\begin{array}{c}
(x: \typebase, \typedestructor, \decomp_1) \propagateparamto \decomp_1'
\quad \ldots \quad
(x: \typebase, \typedestructor, \decomp_n) \propagateparamto \decomp_n'
\end{array}}
{
(x: \typebase, \typedestructor, \bounded{\vdfunc}{\decomp_1, \ldots, \decomp_n})
\propagateparamto
\bounded{\vdfunc}{\decomp_1', \ldots, \decomp_n'}
}
{\textsc{(Tuple-Base)}}

\\ \ \\

\irulelabel
{\begin{array}{c}

\bar{\decomp}_i = \decomp_i[v_i / \typeselector_i(x)] \quad \fresh(v_i) \quad x \not\in \fv(\bar{\decomp}_i) \\

(v_1: \type_1, \typeselector_1 \funccomp \typedestructor, \bar{\decomp}_1) \propagateparamto \decomp_1'
\quad \ldots \quad
(v_n: \type_n, \typeselector_n \funccomp \typedestructor, \bar{\decomp}_n) \propagateparamto \decomp_n'
\end{array}}
{
(x: (\tau_1, \ldots, \tau_n)), \typedestructor, \bounded{\vdfunc}{\decomp_1, \ldots, \decomp_n})
\propagateparamto
\bounded{\vdfunc}{\decomp_1', \ldots, \decomp_n'}
}
{\textsc{(Tuple-Inductive)}}

\\ \ \\

\irulelabel
{\begin{array}{c}
(x: \typebase, \typedestructor, \decomp) \propagateparamto \decomp'
\end{array}}
{
(x: \typebase, \typedestructor, \unbounded{\typeiter}{\decomp}{\vdpred}{\typedestructor_0})
\propagateparamto
\unbounded{\typeiter}{\decomp'}{\vdpred}{\typedestructor_0}
}
{\textsc{(C-Base)}}

\irulelabel
{\begin{array}{c}
(x: \type, \funcid, \decomp) \propagateparamto \decomp'
\end{array}}
{
(X: \typecollection{\type}, \typedestructor, \unboundedapp{\typeiter}{\decomp}{\vdpred}{x \in X})
\propagateparamto
\unbounded{\typeiter}{\decomp'}{\abslambda{x}{\vdpred}}{\typedestructor}
}
{\textsc{(C-Ind)}}

\end{array}
\]
\vspace{-0.1in}
\caption{Function simplification rules (shown for a core subset of the expressions).}
\label{fig:pp}
\vspace{-0.1in}
\end{figure}

%% file: figures/pseudocode.tex
\begin{figure}[!t]
\begin{algorithm}[H]
\caption{Homomorphism verification algorithm}
\label{alg:top-level}
\begin{algorithmic}[1]
\Procedure{\textsc{IsHomomorphism}}{$\dfagg = \abslambda{x}{\agg(\udaf, \init, x)}$}
\vspace{2pt}
\Statex \textbf{Input:} A dataframe aggregation $\dfagg$
\Statex \textbf{Output:} Merge operator  or $\bot$ (if $\dfagg$ is not homomorphic)

\IIf{\Call{ApplyHomRefute}{$\dfagg$}} \Return $\bot$ \label{lst:top-level:step1}\label{lst:top-level:refute1}
\IIf{\Call{ApplyNormRefute}{$\udaf, \init$}} \Return $\bot$ \label{lst:top-level:refute2}

\If{$\neg$\Call{CanApplyDecomp}{$f$}} 
\State \Return \Call{ApplyNormSynth}{$f, \rfs$}\Comment{Attempt normalizer synthesis} \label{lst:top-level:step2}
\EndIf

\State $(\decomp, S) \gets$ (\Call{ApplyDecomp}{$f$}, $\varnothing$) \Comment{Decompose $f$}\label{lst:top-level:decompose}
\State $E \gets$ \Call{ApplyExprConvert}{$\decomp$} \Comment{Convert to canonical form}
\State $\mathsf{refuted} \gets \textbf{false}$

\ForAll{$(\udaf_i, \init_i) \in$ \Call{MatchNormInductive}{$E$}} \label{lst:top-level:step3}
    \State $h_i \gets$ \Call{ApplyNormRefute}{$f_i, \init_i$} $?$ $\bot$ $:$ \Call{ApplyNormSynth}{$f_i, \init_i$} \label{lst:top-level:norm-refute}
    \IIf{$h_i = \bot$} $\mathsf{refuted} \gets \textbf{true}$; \textbf{break} 
    \State $S \gets S \cup \{f_i \rightarrow h_i\}$ \label{lst:top-level:step3-end}
\EndFor

\IIf{$\neg$\textsf{refuted}} \Return \Call{ApplyNormSynth}{$f, \rfs$} \label{lst:top-level:step4}
\State \Return \Call{ApplyNormInductive}{$S, E$} \label{lst:top-level:step5}

\EndProcedure
\vspace{2pt}

\end{algorithmic}
\end{algorithm}
\vspace{-0.2in}
\end{figure}

%% file: impl.tex
\section{Implementation}\label{sec:impl}

We have implemented our proposed technique in a tool called \toolname{} (written in Rust) which  uses the CVC5~\cite{cvc5} synthesizer to solve leaf-level synthesis problems. \toolname{} takes as input a UDAF and outputs  its corresponding merge function (if one exists). The input and output of \toolname{} is implemented in the language from Figure~\ref{fig:dsl-syntax}, but \toolname{} also accepts source programs written in Scala, which \toolname{} converts to its own language using a custom transpiler.


\bfpara{SyGuS encoding for UDAFs.}
While our method tries to decompose the synthesis problem as much as possible, there may still be leaf-level synthesis problems that involve unbounded data structures, which is a challenge for the SyGuS encoding. 
Our implementation deals with  this challenge by modeling lists as sequences in CVC5 and adding list transformation functions like map and filter to the SyGuS grammar. Our implementation similarly models maps as sets of key-value pairs, using the theory of sets supported in CVC5. To ensure determinism and avoid dependence on solver-specific heuristics, we sort the non-terminals in the generated SyGuS grammar alphabetically.  A summary of the SyGuS grammar for leaf-level synthesis is provided in \Cref{appendix:sygus}.

\bfpara{Implementation of refutation rules.} The refutation rules in our calculus disprove the existence of a merge operator by finding inputs that satisfy a certain property.
In our implementation, we use property-based testing (specifically, the Rust implementation of QuickCheck) to perform refutation.
To this end, we provide the logical negation of our refutation rules as properties to be checked during testing.
While we also experimented with a refutation procedure based on  SMT, we found that this approach  can be quite slow due to the presence of unbounded data structures. Thus, our implementation uses testing by default.



\bfpara{Incompleteness of decomposition.}
Due to the incompleteness of decomposition (see \Cref{ex:incomplete-decomp}),  \Cref{alg:top-level}  falls back on syntax-guided synthesis if the decomposed functions are not homomorphic. However, in practice, we found that \emph{most} functions in the decomposition are homomorphic even when not \emph{all} of them are.  Furthermore, we found that the solution for these sub-problems are still useful for solving the overall problem. Our implementation leverages this insight by incorporating solutions to these sub-problems as terminals in the grammar of the SyGuS encoding.

\bfpara{Generalization of \textsc{Norm-Tuple}.}
Recall that the  \textsc{Norm-Tuple} rule in \Cref{fig:witness_construction} tries to decompose a function $f$ as a sequence of functions $f_1, \ldots, f_n$, each of which operates over a single element of the tuple. Our implementation generalizes this rule and allows each $f_i$ to operate over a subset of the elements in the tuple. For instance, using this generalization, a function such as $\abslambda{(a, b, c)}{\abslambda{x}{(a + 1, \ite(c, b + x, b), \vbooltrue)}}$ can be decomposed into the following two functions:
\begin{equation*}
    \udaf_1 = \abslambda{a}{\abslambda{x}{a + 1}} \quad\quad \udaf_2 = \abslambda{(b, c)}{\abslambda{x}{(\ite(c, b + x, b), \vbooltrue)}}
\end{equation*}

\bfpara{Transpiler from Scala.} While \toolname{} can take Scala source code as input, it first transpiles Scala code to its own intermediate representation (see \Cref{fig:dsl-syntax}). The Scala-to-\toolname{} transpiler follows a syntax-directed translation process that systematically rewrites UDAFs into \toolname{}'s IR. It first maps accumulator state representations by converting Scala primitive types into \toolname{} primitive types, while standard collections like \texttt{List}, \texttt{Map}, and \texttt{Set} are translated to Ink's collection primitives. Since most Scala UDAFs already exhibit a functional structure,  transpilation lends itself to straightforward syntax-directed translation for most Scala UDAFs implemented in frameworks like Spark and Flink.  However, for UDAFs using third-party libraries or custom types,  the user needs  to provide mappings from these to \toolname{}'s built-in collection types.  The transpiler from Scala to the \toolname{} IR is implemented in Python.

%% file: eval.tex
\section{Evaluation}

We now describe our experiments that are designed to answer the following research questions:

\begin{enumerate}[label = \textbf{RQ\arabic*.}]
    \item How does \tool{} compare against relevant baselines for merge operator synthesis? 
    \item How does \tool{} compare with other tools for refuting homomorphisms?
    \item How important are the core ideas (deduction, decomposition) for merge function synthesis?
    \item How important are the refutation rules?
\end{enumerate}

\begin{table}[t]
\centering
\begin{minipage}[b]{0.46\linewidth}
\caption{AST Statistics.}
\vspace{-0.10in}
\label{tab:ast_stats}
\centering
\begin{tabular}{|l|c|c|}
    \hline
    \textbf{Metric}          & \textbf{UDAF} & \textbf{Merge} \\ \hline
    Avg AST          & 30.6          & 21.3           \\ \hline
    Median AST           & 27.5          & 22.0           \\ \hline
    Max AST              & 120.0         & 106.0           \\ \hline
\end{tabular}
\end{minipage}
\hspace{0.04\linewidth} 
\begin{minipage}[b]{0.46\linewidth}
\caption{Other metrics.}
\vspace{-0.10in}
\label{tab:other_metrics}
\centering
\begin{tabular}{|l|c|}
    \hline
    \textbf{Metric}        &               \\ \hline
    Tuples                 & $72.0\%$      \\ \hline
    Collections            & $42.0\%$      \\ \hline
    Conditionals           & $50.0\%$      \\ \hline
\end{tabular}
\end{minipage}
\vspace{-0.18in}
\end{table}

\bfpara{Sources of benchmarks.}
To answer these questions, we sampled a set of approximately 100 benchmarks from real-world GitHub repositories, focusing on implementations of UDAFs for Apache Spark and Flink.
These benchmarks represent a mix of UDAFs that span a diverse set of domains such as telemetry, finance, geospatial and raster analytics, and machine learning. To ensure that our evaluation focuses on non-trivial UDAFs, we perform further filtering of the collected UDAFs  by retaining only those functions that satisfy the  following two criteria: First, the UDAF must contain at least 10 LOC, and, second, it should involve control flow or a non-trivial type (namely, collection or tuple with at least three elements). After filtering easy benchmarks that do not meet this criteria, we obtain a total of 50 benchmarks, of which 45 are dataframe homomorphisms. 
\Cref{tab:ast_stats} provides statistics about the size of these UDAFs and their corresponding merge operator in terms of average, median, and maximum AST size. Additionally, \Cref{tab:other_metrics} reports the percentage of UDAFs that contain tuples, collections, and conditionals.

\bfpara{Baselines.} To evaluate the effectiveness of our approach, we compare our method against two baselines. The first baseline is CVC5, a state-of-the-art SyGuS solver that provides support for data structures like tuples and sets. 
Our other  baseline is a state-of-the-art synthesizer called \parsynt{}~\cite{parsynt_pldi21}  for \emph{divide-and-conquer} algorithms---notably, \parsynt{} also aims to synthesize merge operators that can be used for parallelization. 
Additionally, we tried to compare \toolname{} against \autolifter{}, which is another  synthesizer for divide-and-conquer parallelism. However, the implementation of \autolifter{} assumes that the UDAF output is a scalar value. Since this assumption does not hold for our benchmarks, we were unable to perform an empirical evaluation against \autolifter{}.

\bfpara{Experimental setup.}
All  experiments are conducted on a machine with an AMD Ryzen 9 7950X3D CPU and 64 GB of  memory, running  NixOS 24.11. We use a 10 minute time limit for each benchmark.

\subsection{Evaluation of Merge Operator Synthesis}

To answer our first research question, we run \toolname and both baselines  on the 45 homomorphic UDAFs and evaluate their ability to synthesize merge operators. The results of this evaluation are shown in \Cref{fig:synthesis-results}  where the $x$-axis
shows cumulative running time and the $y$-axis represents the percentage of benchmarks solved.  \toolname{} is able to successfully synthesize a merge operator for all but 2 benchmarks, resulting in a success rate of 95.6\%. In contrast, \parsynt{} and CVC5 are able to synthesize merge operator for 42.2\% and 28.9\% of the benchmarks respectively. All of the benchmarks solved by \parsynt{} and CVC5 are also solved by \toolname{}. 
Additionally, we note that the synthesis time for \toolname{} is 6.2 seconds per benchmarks on average. Among benchmarks that can be solved by both \toolname{} and \parsynt{} (resp. {\sc CVC5}), \toolname is $2.2\times$ (resp. $28.3\times$) faster on average.

\begin{figure}[t]
\vspace{-0.2in}
\centering
\begin{minipage}[b]{0.46\linewidth}
    \centering
    \includegraphics[width=\textwidth]{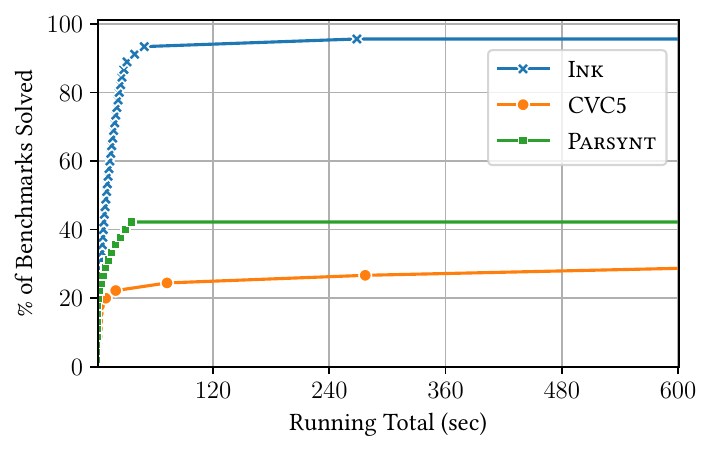}
    \vspace{-0.35in}
    \caption{Comparison between \toolname{} and baselines for merge operator synthesis.}\label{fig:synthesis-results}
\end{minipage}
\hspace{0.04\linewidth} 
\begin{minipage}[b]{0.46\linewidth}
    \centering
    \includegraphics[width=\textwidth]{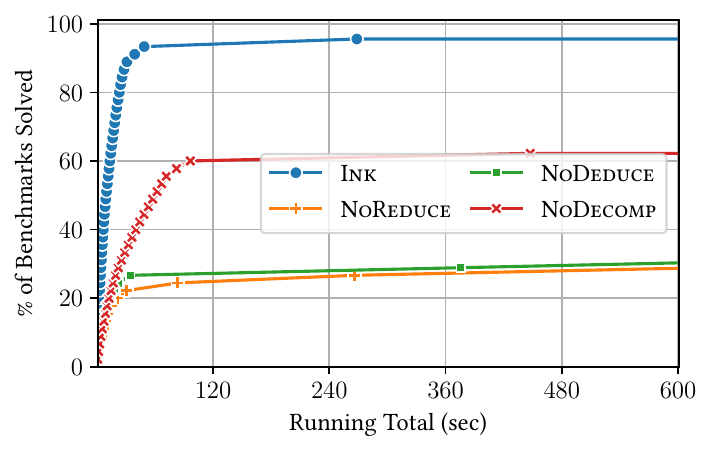}
    \vspace{-0.35in}
    \caption{Comparison between \toolname{} and its ablations for merge operator synthesis.}\label{fig:synthesis-ablation}
\end{minipage}
 \vspace{-0.18in}
\end{figure}

\begin{figure}
\small
\centering
\begin{minted}
[
fontsize=\footnotesize,
escapeinside=||,
numbersep=5pt,
]
{scala}
case class InputData(key: String, value: Int)
case class BufferData(key: String, sum: Int, count: Int)

object AvgTemperatureAggregator extends Aggregator[InputData, BufferData, ...] {
  override def zero: BufferData = BufferData("", 0, 0)

  override def reduce(buffer: BufferData, input: InputData): BufferData = {
    BufferData(input.key, buffer.sum + input.value, buffer.count + 1)
  } }
\end{minted} 
\vspace*{-0.15in}

\caption{A Scala \spark{} UDAF with incorrect user-provided merge (see Figure~\ref{fig:inconsistent-gt-user})}
\label{fig:inconsistent-gt}
\vspace{-0.15in}
\end{figure}

\begin{figure}[t]
\centering
\begin{minipage}[b]{0.46\linewidth}
    \centering
\begin{minted}
[
fontsize=\footnotesize,
escapeinside=||,
numbersep=5pt,
]
{scala}
def merge(b1, b2): BufferData = {
  BufferData(
    b2.key,
    b1.sum + b2.sum,
    b1.count + b2.count) }
\end{minted} 
    \caption{User-provided incorrect merge.}\label{fig:inconsistent-gt-user}
\end{minipage}
\hspace{0.04\linewidth} 
\begin{minipage}[b]{0.46\linewidth}
    \centering
\begin{minted}
[
fontsize=\footnotesize,
escapeinside=||,
numbersep=5pt,
]
{scala}
def merge(b1, b2): BufferData = {
  BufferData(
    if (b2.count > 0) b2.key else b1.key,
    b1.sum + b2.sum,
    b1.count + b2.count) }
\end{minted} 
    \caption{\toolname{}-synthesized solution.}\label{fig:inconsistent-gt-ink}
\end{minipage}
\vspace{-0.25in}
\end{figure}

\bfpara{Qualitative Analysis for \toolname{}.} 
Of the 43 benchmarks \toolname{} solves, we found that 36 of the synthesized programs are semantically equivalent to the developer-provided merge operator.  
In all cases where \toolname{}'s results are semantically equivalent to the user-provided one, the synthesized merge function also has the same time and space complexity.
However, there are also  cases where the synthesized merge function  differs  from its human-written counterpart,   revealing potential bugs. 
For instance, \Cref{fig:inconsistent-gt} shows a UDAF along with its corresponding human-written merge function in \Cref{fig:inconsistent-gt-user}.  Although this merge function appears reasonable, it produces incorrect results when merging an accumulated state with the initial state—that is, the default accumulator produced by the UDAF's initializer. As illustrated in \Cref{fig:incorrect-merge-illu}, the merge function incorrectly overwrites the key field with the default value, even though the initial state should have no effect. This behavior violates the homomorphism property and can lead to incorrect or non-deterministic results during distributed execution, as the final output may depend on how data is partitioned and merged. In contrast, \toolname{} synthesizes the correct merge operator shown in \Cref{fig:inconsistent-gt-ink}, which preserves the intended semantics of the UDAF.

\begin{wrapfigure}{r}{0.45\textwidth}
  \vspace{-0.3in}
  \centering
  \small
  \begin{tcolorbox}[colback=gray!5!white, colframe=gray!50!black, 
                   boxrule=0.5pt, arc=2pt, left=2pt, right=2pt, top=2pt, bottom=2pt,
                   width=0.44\textwidth]
    \vspace{-0.1in}
    \begin{align*}
    \shortintertext{\textbf{Inputs and Intermediate Outputs}}
        \df &= \texttt{[("key", 5)]} \\
        \prog(\df) &= \texttt{("key", 5, 1)} \\
        \prog(\emptylist) &= \texttt{("", 0, 0)} \\
    \shortintertext{\textbf{Mismatched Merge Result}}
        \prog(\df ~\pp~ \emptylist) &= \texttt{(\correctcode{"key"}, 5, 1)} \\
        \texttt{merge}(\prog(\df), \prog(\emptylist)) &= \texttt{(\wrongcode{""}, 5, 1)}
    \end{align*}
  \end{tcolorbox}
  \vspace{-0.15in}
  \caption{Illustration of incorrect merge.}\label{fig:incorrect-merge-illu}
  \vspace{-0.1in}
\end{wrapfigure}


\bfpara{Failure analysis for baselines.} We manually inspected the failure cases for both baselines. For {\sc CVC5}, we observe an inverse correlation between the size of the required merge operator and {\sc CVC5}'s ability to find it. In fact, the 13 benchmarks that {\sc CVC5} can solve are the smallest benchmarks of the 45 in terms of AST size. On the other hand, \parsynt{} is able to solve more of the benchmarks that involve tuples compared to {\sc CVC5} but it struggles with benchmarks where the accumulator state is a collection. Interestingly, there are only four benchmarks that \emph{both} {\sc CVC5} and \parsynt{} can solve, indicating that these tools have different failure modes.

\bfpara{Failure analysis for \toolname{}.} As mentioned earlier, there are two benchmarks that \toolname{} fails to solve; both are due to time-outs when solving a leaf-level synthesis problem. One of them requires synthesizing a non-linear expression as part of the merge operator. Since SMT solvers typically struggle with non-linear operations, failure on this benchmark is not very surprising. The second failure is more surprising---in fact, if we change the order of non-terminals in our SyGuS grammar, then {\sc CVC5} is able to solve the same leaf-level synthesis problem. \\


\idiotbox{RQ1}{Among the 45 homomorphic UDAFs, \toolname{} can successfully synthesize merge operators for 43 of them ($95.6\%$), taking 6.2 seconds per benchmark on average. In comparison, the  two baselines ({\sc CVC5} and \parsynt{})  synthesize merge operators for $28.9$-$42.2\%$ of the same benchmarks.}

\begin{figure}
\small
\centering
\begin{minted}
[
fontsize=\footnotesize,
escapeinside=||,
numbersep=5pt,
]
{scala}
case class ClickEventAggregate(
  userId: Int = 0, eventCount: Int = 0,
  eventCountWithOrderCheckout: Int = 0, departmentsVisited: Set[String] = Set())
case class ClickEvent(userid: Int, productType: String, eventType: String)

object ClickstreamAggregator extends Aggregator[ClickEvent, ClickEventAggregate, ...] {
  override def zero: ClickEventAggregate = ClickEventAggregate()
  override def reduce(accumulator: ClickEventAggregate, value: ClickEvent): ClickEventAggregate = {
    if (value.productType.nonEmpty && value.productType != "N/A") {
      accumulator.eventCount += 1
      val departmentsVisited = accumulator.departmentsVisited
      departmentsVisited.add(value.productType)
      accumulator.departmentsVisited = departmentsVisited
    }
    if (accumulator.userId == 0) { accumulator.userId = value.userid }
    if (value.eventType == "order_checkout") {
      accumulator.eventCountWithOrderCheckout = accumulator.eventCount
    }

    accumulator } }
\end{minted} 
\vspace*{-0.1in}

\caption{A non-homomorphic Scala \spark{} UDAF.}
\label{fig:nonhomomorphic-example}
\end{figure}

\begin{figure}
\small
\centering
\begin{minted}
[
fontsize=\footnotesize,
escapeinside=||,
numbersep=5pt,
]
{scala}
  override def merge(acc1: ClickEventAggregate, acc2: ClickEventAggregate): ClickEventAggregate = {
    acc1.copy(
      eventCount = acc1.eventCount + acc2.eventCount,
      eventCountWithOrderCheckout = acc1.eventCountWithOrderCheckout + acc2.eventCountWithOrderCheckout,
      departmentsVisited = acc1.departmentsVisited ++ acc2.departmentsVisited) }
\end{minted} 
\vspace*{-0.1in}

\caption{The user-provided merge function to \texttt{ClickstreamAggregator} from Figure~\ref{fig:nonhomomorphic-example}}
\label{fig:nonhomomorphic-example-merge}
\end{figure}

\subsection{Evaluation on Non-Homomorphic UDAFs}


Of the 50 benchmarks used in our evaluation, 5 are non-homomorphic. Notably, although the source files for these benchmarks define merge operators, these operators are either buggy or depend on undocumented assumptions—such as the absence of certain values in the dataframe. \toolname{} successfully refutes the existence of a valid merge operator for all five non-homomorphic benchmarks, with a median refutation time of 0.6 seconds and an average of 1.4 seconds.  
{\Cref{fig:nonhomomorphic-example} shows a UDAF that \toolname{} proves to be non-homomorphic. For this UDAF (abbreviated as \texttt{CSA} below), there is in fact \emph{no} \texttt{merge} function that can satisfy the required correctness property:
\[
\forall \df_1, \df_2. \ \ \texttt{CSA}(\df_1 ~\pp{}~ \df_2) = \texttt{merge}( \texttt{CSA}(\df_1),  \texttt{CSA}(\df_2))
\]
 However, the user nevertheless provides the merge function shown in \Cref{fig:nonhomomorphic-example-merge}, which is incorrect as illutsrated through the inputs shown  in \Cref{fig:nonh-inputs}. 
Such an incorrect merge function would lead to inconsistent results across different partitioning of the input. }

\begin{figure}[t]
  \centering
  \small
  \begin{tcolorbox}[colback=gray!5!white, colframe=gray!50!black,
    boxrule=0.5pt, arc=2pt, left=3pt, right=3pt, top=2pt, bottom=2pt,
    width=0.9\linewidth]
    \vspace{-0.1in}
    \begin{align*}
\shortintertext{\textbf{Inputs and Intermediate Outputs}}
        \df_1 &= \texttt{[(5, "product", "add\_cart")]} \\
        \df_2 &= \texttt{[(0, "N/A", "order\_checkout")]} \\
    \texttt{$\prog(\df_1)$} &= \texttt{(5, 1, 0, \{"product"\})} \\
    \texttt{$\prog(\df_2)$} &= \texttt{(0, 0, 0, \{\})} \\
\shortintertext{\textbf{Illustration of incorrect merge}}
    \texttt{$\prog(\df_1 ~\pp~ \df_2)$} &= \texttt{(5, 1, \correctcode{1}, \{"product"\})}\\
        \texttt{merge}(\prog(\df_1), \prog(\df_2)) &= \texttt{(5, 1, \wrongcode{0}, \{"product"\})}
    \end{align*}
  \end{tcolorbox}
  \vspace{-0.8em}
  \caption{Illustration of why the \texttt{merge} function from Figure~\ref{fig:nonhomomorphic-example-merge} is incorrect}
  \label{fig:nonh-inputs}
  \vspace{-1em}
\end{figure}

Among the two baselines, {\sc CVC5} also includes refutation capabilities and, in principle, it can return \textsf{Infeasible} for problems it proves unrealizable. CVC5 refutes 2 out of 5 non-homomorphic UDAFs with an average refutation time of 8.2 seconds.
By contrast, \parsynt{} attempts to construct a \emph{homomorphic lifting} of the input function if it is non-homomorphic. A homomorphic lifting aims to transform a non-homomorphic function into an equivalent form that satisfies the homomorphism property. For the 5 non-homomorphic functions in our benchmark set, \parsynt{} either times out or produces a merge operator with an empty body. \\

\idiotbox{RQ2}{\toolname is able to refute the existence of a merge operator for all five non-homomorphic UDAFs, whereas the baselines refute at most 2.}

\subsection{Ablation Study for Merge Operator Synthesis}

To address our third research question, we consider several ablations of \toolname{}. The first one is {\bf \noreduce{}}, which does not reduce the merge operator synthesis problem to that of finding a normalizer for the accumulator function. In other words, this ablation uses \Cref{eq:df-homo} as the specification instead of \Cref{thm:norm-sound} and is therefore expected to have behavior similar to the {\sc CVC5} baseline. The second ablation is {\bf \nodeduce{}}, which does not use the deductive synthesis rules for normalizer synthesis. In other words, this ablation invokes CVC5 with the normalizer specification; however it does not utilize the {\sc Norm-Coll} and {\sc Norm-Tuple} rules for deductive synthesis. The third ablation is {\bf \nodecomp{}}, which does not perform the type-directed decomposition method described in \Cref{ssec:expr_decompose}. Hence, this ablation can only leverage the {\sc Norm-Tuple} and {\sc Norm-Coll} rules  if the original expression has the required syntactic form.

The results of this ablation study are presented in \Cref{fig:synthesis-ablation}, where the $x$-axis shows cumulative running time and $y$-axis represents the number of benchmarks solved. As expected, \noreduce{} has the same performance as CVC5. \nodeduce{} performs slightly better than \noreduce{}, but it can still only synthesize $33.3\%$  of the benchmarks. Finally, \nodecomp{} is the best-performing ablation but solves $34.9\%$ fewer benchmarks compared to \toolname and takes longer to do so. \\

\idiotbox{RQ3}{Without deductive synthesis and decomposition, the synthesis capability of \toolname degrades considerably, solving $34.9$-$69.8\%$ fewer benchmarks.}

\subsection{Ablation Study for Refutation}

Finally, we perform an ablation study to evaluate the usefulness of our refutation rules. For this experiment, we  consider the five non-homomorphic UDAFs and compare \toolname against {\bf \norefute{}}, which is a version of \toolname that does not utilize the refutation rules in our calculus. This ablation is able to refute 2 of the 5 non-homomorphic UDAFs but takes $1.5\times$ as long on average to do so. \\ 

\idiotbox{RQ4}{The ablation of \toolname that does not leverage the refutation rules fails to refute 3 of the 5 non-homomorphic benchmarks and takes $1.5\times$ as long. }

%% file: limitations.tex
\section{Limitations} \label{sec:limitations}
In this section, we discuss some of the main limitations of the proposed approach.
First, our problem statement is defined in terms of a functional IR, which means that the UDAF needs to be expressible in this IR. However, our DSL is designed to capture the essential structure of UDAFs as they are implemented in distributed data-processing frameworks like Apache \spark{} and \flink{}, and, empirically, we found that all benchmarks sampled from GitHub can be expressed in our DSL.
Second, our problem statement  assumes that an input UDAF processes a dataframe in a single pass without random access. This assumption aligns with the UDAF frameworks of both traditional databases, such as PostgreSQL~\cite{postgresql}, MySQL~\cite{mysql}, SQL Server~\cite{sqlserver}, and Oracle~\cite{oracle}, and big data systems, such as \spark{} and \flink{}.  

%% file: related.tex
\section{Related Work}\label{sec:related}

\bfpara{Homomorphisms.} 
List homomorphisms have long been a core strategy for transforming sequential operations into parallelizable ones, especially in functional programming~\cite{bird1987,chin1996,cole1993}. Hu et al.~\cite{hu1996,hu1997} introduced systematic techniques for deriving list homomorphisms from recursive functions using methods like tupling and fusion, and they also demonstrated how \emph{almost homomorphic} functions can be converted into fully homomorphic ones to support efficient parallel execution. Gibbons’ third list-homomorphism theorem~\cite{gibbons1996} established that a function qualifies as a list homomorphism if it can be expressed as both a \texttt{foldr} and a \texttt{foldl}, while Mu and Morihata~\cite{pearl} extended this theorem to tree structures.
Building on this foundation, our work explores homomorphisms in the context of UDAFs, but adopts a more local perspective. Whereas prior characterizations, such as Gibbons’ theorem, analyze global properties of the entire function, our calculus shows that an aggregation is a homomorphism if and only if its row-level accumulator admits a normalizer satisfying a generalized commutativity condition. This formulation reduces the synthesis of a global merge operator to the more tractable task of identifying a suitable normalizer for the accumulator. In doing so, it generalizes classical associativity: while associativity emerges as a special case when the accumulator operates uniformly over identical types, our framework allows accumulator state and input-row types to differ, making it applicable to a broader range of real-world aggregations.
Another related work is that of Cutler et al.~\cite{cutler}, which develops a type stream processing calculus that ensures homomorphism by construction. In contrast, our work verifies and synthesizes merge operators for arbitrary UDAFs written in general-purpose languages, where homomorphism is neither assumed nor guaranteed.

\bfpara{Synthesis for parallelization \& incremental computation.}
Program synthesis aims to automatically generate programs that satisfy a given specification, such as input-output examples~\cite{FunctionalSynthesisLambdaSquared,Gulwani2011FlashFill}, logical constraints~\cite{FunctionalSynthesisBurst}, or reference implementations~\cite{revamp,migrator,verified-lifting}. In our case, the merge operator synthesis problem involves both a logical specification (dataframe homomorphism) and a reference implementation (UDAF). While synthesis approaches are generally classified as inductive~\cite{sketch,sygus} or deductive~\cite{deductiveSynthesisMannaWaldinger1980}, our approach combines both, using SyGuS solvers for local synthesis and deductive synthesis to consolidate sub-solutions. 
Several related works, including \Synduce~\cite{Synduce}, \ParSynt~\cite{parsynt_pldi17}, and \AutoLifter~\cite{autolifter_toplas24}, focus on synthesis for parallel computation. \Synduce primarily focuses on recursive procedures, and, like our method, it has the ability to refute the existence of a solution.  \ParSynt and \AutoLifter target divide-and-conquer parallelism through synthesis of join operators.  The focus of \ParSynt is to introduce auxiliary accumulators to lift non-parallelizable loops into parallelizable ones.
  Similar to our approach, \AutoLifter also employs a form of decomposition, but their focus is on decomposing the \emph{specification}. Neither \AutoLifter not \ParSynt  has a mechanism for proving unrealizability, and none of these techniques deal with challenges that arise in the context of synthesizing merge operators for UDAFs. 
Other works, like Superfusion (SuFu)~\cite{SuFu} and MapReduce synthesis~\cite{mapReduceSynthesis}, also emphasize optimizing computations. SuFu eliminates intermediate data structures in functional programs,  and MapReduce synthesis~\cite{mapReduceSynthesis}  uses higher-order sketches to generate mappers and reducers from input-output examples.
Finally, frameworks like {\sc Bellmania}~\cite{DpSynthesis} and {\sc Opera}~\cite{opera} combine inductive and deductive synthesis to achieve incremental computation. {\sc Bellmania} targets dynamic programming algorithms through recursive call generation, while {\sc Opera} transforms batch programs into streaming versions by synthesizing auxiliary states. In contrast, our approach focuses on the  correct synthesis of merge operators for homomorphic dataframe aggregations.

\bfpara{Optimizing user-defined functions.} 
There is a large body of work on optimizing user-defined functions (UDFs). Some methods translate UDFs into SQL operators to leverage traditional relatonal databases and their query optimizers ~\cite{QBS:conf/pldi/CheungSM13,EqSQL:sigmod/EmaniRBS16,Decorrelation:icde/Simhadri0CG014,Froid:pvldb/RamachandraPEHG17,Aggify:sigmod/GuptaP020,zhang2021udf,zhang2023automated}. In the context of big data systems,  prior work has used static analysis to optimize UDFs, enabling predicate pushdown~\cite{AutomaticMR:pvldb/JahaniCR11}, efficient data communication~\cite{Shuffle:nsdi/ZhangZCFGLLLZZ12}, and computation sharing across UDFs~\cite{consolidate}. There is also prior work on optimizing Python-based data science programs by employing predefined rewrite rules~\cite{Dias:journals/pacmmod/BaziotisKM24} and SQL-like  representations~\cite{HiPy}. 
A recent paper uses program synthesis and verification 
in this context to support predicate pushdown in data science programs~\cite{MagicPush:journals/pacmmod/YanLH23}. 
 However, these methods do not address the automatic synthesis of merge functions for homomorphic UDAFs, which is necessary for incremental and parallel execution.

\bfpara{Parallel and incremental execution for data analytics.}
Parallel and incremental execution has long been critical in  data analytics frameworks~~\cite{Spark:nsdi/ZahariaCDDMMFSS12,Dryad,ParallelJoin:sigmod/ManciniKCMA22,Gupta93IVM,Larson86IVM,DBToaster,IDIVM,FIVM,nikolic2016win,LINVIEW,Yannakakis}; however, existing work primarily focuses on optimizing relational operators such as joins and built-in aggregates (e.g., max, average). There is one prior work that  proposes a DSL, inspired 
by digital signal processing, for incremental execution; 
however, its applicability remains limited to this  DSL~\cite{DBSP:pvldb/BudiuCMRT23}.  To the best of our knowledge, no prior technique applies to UDAFs that involve complex logic or data structures.

\bfpara{Dataframe model.} The dataframe model originated in the S Language~\cite{DataframeInS} and gained popularity through R~\cite{RLang}, later becoming central to Pandas~\cite{Pandas}, a widely used Python library for data analysis. It is now integral to systems like Apache Spark~\cite{SparkDF} and Snowflake~\cite{SnowflakeDF}. Unlike the relational model, dataframes have ordered rows and support complex column types, such as lists or arrays. Since user-defined aggregation functions (UDAFs) in many systems process rows sequentially and handle complex types, our homomorphism calculus adopts  dataframes as its underlying model.

%% file: conclusion.tex
\section{Conclusion}

We presented a calculus for reasoning about the homomorphism property of user-defined aggregation functions. The key idea of our calculus is to re-formulate the homomorphism property in terms of a generalized commutativity condition between the merge operator and the accumulator function of the aggregation. Based on this formulation, our proposed method decomposes the original problem into independent sub-problems in a type-directed fashion and uses deductive synthesis to combine the results. Our experimental evaluation on 50 real-world UDAFs shows that our proposed algorithm can solve 96\% of these benchmarks, substantially outperforming state-of-the-art synthesizers as well as its own ablations.

%% file: data-avail.tex
\section{Data-Availability Statement}

Our artifact, including the benchmark suite and a copy of \toolname{}'s output, can be found on Zenodo~\cite{wang_2025_16915406}.

%% file: acks.tex
\begin{acks}
We thank the anonymous reviewers for their thoughtful and constructive feedback.
This work was conducted in a research group supported by NSF awards CCF-1762299, CCF-1918889, CNS-1908304, CCF-1901376, CNS-2120696, CCF-2210831, CCF-2319471, CCF-2422130, and CCF-2403211 as well as a DARPA award under agreement HR00112590133.
\end{acks}

%% file: proof.tex
\section{Proofs}

\subsection{Proof of Non-Existence of Normalizer for \Cref{ex:norm-nonexist}}\label{proof:ex-norm-nonexist}

Consider the following function $f: \mathbb{N} \times \mathbb{N} \rightarrow \mathbb{N}$ which is a right action of $\mathbb{N}$ on $\mathbb{N}$:
\[
f(y,z) = \left \{ 
\begin{array}{ll}
0 & \mathsf{if} \   y = z \\
z & \mathsf{otherwise}
\end{array}
\right .
\]
We prove that a normalizer of $f$ cannot exist. For contradiction, suppose a normalizer $h$ of $f$ existed. Then, it would need to satisfy the following axiom for all $x, y, z$:
\[
h(x, f(y, z)) = f(h(x, y), z)
\]
Since this equality needs to hold for all $x, y, z$, consider the scenario where $y \neq z$. Then, we get $
h(x, z) = f(h(x, y), z)$.
Using the definition of $f$, we obtain:
\[
h(x, z) =  
\left \{ 
\begin{array}{ll}
0 & \mathsf{if} \   h(x,y) = z \\
z & \mathsf{if} \   h(x,y) \neq z 
\end{array}
\right .
\]
However, this violates the definition of a function. To see why, let $x = a, y= b, h(a, b) = c$. We get:
\[
h(a, z) =  
\left \{ 
\begin{array}{ll}
0 & \mathsf{if} \    c = z \\
z & \mathsf{if} \   c \neq z 
\end{array}
\right .
\]
Clearly, this implies $h(a, c) = 0$.
We consider two cases:
\begin{itemize}
    \item $c=0$. Now, consider the term $h(a, h(a, d))$ for any $d \neq 0$. From the second part of the definition, we get $ h(a, d) = d $ (since $c=0$ and $d \neq 0)$; hence, $h(a, h(a, d)) = h(a, d) = d$ (since $d \neq 0,  c= 0)$. But now since $h(a, d) = d$, the first part of the definition applies to $h(a, h(a,d))$, so we get $h(a, h(a, d)) = 0$. Since we derived both $h(a, h(a, d)) \neq 0 $ (since $d \neq 0$) and $h(a, h(a, d)) = 0$, we get a contradiction.
    \item $c \neq 0$. From earlier, we have $h(a,c) = 0$; thus, $c \neq h(a,c)$. Consider the term $h(a, h(a, c)) $. Since $h(a, c) = 0$, from the second part of the definition, we get $h(a, c) = c$, so $h(a, h(a, c)) = h(a, c) = 0$.  But since $h(a, c) =0$, $h(a, c) = c$, and yet $c \neq 0$, this contradicts our assumption that $h$ is a function. 
    \end{itemize}

\subsection{Proof of \Cref{thm:norm-sound}}

\begin{theorem}
Let $\dfagg = \abslambda{x}{\agg(\udaf, \init, x)}$  be a program where $\udaf: \tau_r \times \tau \rightarrow \tau_r$ is a right action of $\tau$ on $\tau_r$, and let $h$ be a normalizer of $f$ satisfying $\forall s \in \tau_r. h(s, \init) = s$. Then, $\dfagg$ is a  homomorphism.
\end{theorem}

Let $\dfagg = \abslambda{x}{\agg(\udaf, \init, x)}$ be an arbitrary program where $\udaf: \tau_r \times \tau \rightarrow \tau_r$. 
Suppose that $\norm$ is a normalizer of of $\udaf$ satisfying $\forall s \in \type_r. \norm(s, \init) = s$. We show that $\dfagg$ is a dataframe homomorphism with merge operator $\norm$.

Let $X$ and $Y$ be arbitrary dataframes. We induct on the number of rows $m$ of the dataframe $Y$.

\bfpara{Base Case: $m=0$.}
By definition, we have
\begin{equation*}
    \begin{array}{rllr}
         \agg(f, \init, X \dfconcat Y) & = &\agg(f, \init, X \dfconcat \emptylist) &\\ 
                    & = &\agg(f, \init, X) &\\ 
                    & = & h(\agg(f, \init, X   ), \init) & \sidecond{Since $h(s, \init) = s$}\\
                    & = & h(\agg(f, \init, X), \agg(f, \init, \emptylist)) & \sidecond{By definition of $\agg$} 
    \end{array}
\end{equation*}

\bfpara{Inductive Step.} Let $\dfrows_Y = [y_1,...,y_m]$ be the rows of $Y$. Assume that for $X, Y$,
$$\agg(f, \init, X \dfconcat Y) = h(\agg(f, \init, X), \agg(f, \init, Y))$$

Consider $Y' = Y \dfconcat [y_{m+1}] = [y_1, ..., y_m, y_{m+1}],$ where $y_{m+1}$ is an arbitrary row. Then
\begin{equation*}
    \begin{array}{rlr}
        \agg(f, \init, \hspace*{-0.1in}& X \dfconcat Y')  =  \agg(f, \init, X \dfconcat [y_1, ..., y_m, y_{m+1}]) &\\
       & =  \agg(f, \init, X \dfconcat Y \dfconcat [y_{m+1}]) &  \sidecond{By definition of $\dfconcat$}   \\ 
       & =  f(\agg(f, \init, X\dfconcat Y), y_{m+1}) &  \sidecond{By definition of $\agg$}    \\ 
       & =  f(h(\agg(f, \init, X), \agg(f, \init, Y)), y_{m+1}) & \sidecond{By IH}       \\
       & =  h(\agg(f, \init, X), f(\agg(f, \init, Y), y_{m+1})) & \sidecond{By commutativity} \\
       & =  h(\agg(f, \init, X), \agg(f, \init, Y')) &   \sidecond{By definition of $\agg$}     
    \end{array}
\end{equation*}

\subsection{Proof of \Cref{thm:norm-complete}} \label{proof:thm-norm-complete}

Suppose that $\dfagg = \lambda x. \agg(f, \init, x)$ is a surjective dataframe homomorphism, and $f: \type_r \times \type \rightarrow \type_r$ is a user-defined function. Since $\dfagg$ is a dataframe homomorphism, there exists a merge operator $\dfmerge$ by \Cref{def:df-concat}. Let $\norm(x, y) = x \dfmerge y$ be the desired normalizer.

\bfpara{$\norm{}$ is a normalizer of $\udaf$.}
Suppose $a, b \in \type_r$ and $c \in \type$. Since $\dfagg$ is a surjective, there exists dataframe $X$ (resp. $Y$) such that $\dfagg(X) = a$ (resp. $\dfagg(Y) = b$). Note that
\begin{align*}
    \udaf(\norm(\dfagg(X), \dfagg(Y)), c) &= \udaf(\dfagg(X \dfconcat Y), c) = \dfagg(X \dfconcat Y \dfconcat [c]) \\
    \norm(\dfagg(X), \udaf(\dfagg(Y), c)) &= \norm(\dfagg(X), \dfagg(Y \dfconcat [c])) = \dfagg(X \dfconcat Y \dfconcat [c]),
\end{align*}
which implies $\forall a, b, c. \udaf(\norm(a, b), c) = \norm(a, \udaf(b, c))$.

\subsection{Proof of \Cref{thm:norm-unique}} \label{proof:thm-norm-unique}
Let $\dfagg = \lambda x. \agg(f, \init, x)$ be an arbitrary surjective dataframe homomorphism where $f: \type_r \times \type \rightarrow \type_r$ is a user-defined function.

\bfpara{Existence of normalizer $h$ of $f$ satisfying  $\forall s \in \tau_r. h(s, \init) = s$.}
Since $\dfagg$ is a dataframe homomorphism, there exists a merge operator $\dfmerge$ by \Cref{def:df-concat}. Let $\norm(x, y) = x \dfmerge y$ be the desired normalizer.

Suppose $s \in \type_r$. Since $\dfagg$ is surjective, there exists $X$ such that $\dfagg(X) = s$. Then,
\begin{align*}
    \norm(s, \init) = \norm(\dfagg(X), \dfagg(\emptylist)) = \dfagg(X \dfconcat \emptylist) = \dfagg(X) = s.
\end{align*}

\bfpara{$\norm$ satisfying $\forall s \in \tau_r. h(s, \init) = s$ is unique.}
Assume by contradiction that there are two semantically different normalizers of $\udaf$, $h_1$ and $h_2$, such that $\forall s \in \tau_r. h_1(s, \init) = s \land h_2(s, \init) = s$.
Then, there exists $a, b \in \type_r$ such that $h_1(a, b) \neq h_2(a, b)$. Note that there exists dataframe $X$ and $Y$ such that $a = \dfagg(X)$ and $b = \dfagg(Y)$ because of the surjectivity condition $\type_r = \range{\dfagg}$.

By \Cref{thm:norm-sound}, $\dfagg$ is a data frame homomorphism and both $h_1$ and $h_2$ are valid merge operator of $\dfagg$. However, we have
\begin{align*}
    h_1(a, b) = h_1(\dfagg(X), \dfagg(Y)) = \dfagg(X \dfconcat Y) = h_2(\dfagg(X), \dfagg(Y)) = h_2(a, b),
\end{align*}
which is a contradiction.

\subsection{Proof of \Cref{thm:soundness_completeness_hom_synthesis_rules}}

\begin{lemma}\label{lem:soundness_and_completeness_horiz}
Let $\dftrans$ be a data transformation expression with free variable $x$. Then, $\dftrans[(x_1 \dfconcat x_2)/x] = \dftrans[x_1/x] \dfconcat \dftrans[x_2/x]$ if and only if $\dftrans[(x_1 \dfconcat x_2)/x] \horizdecompto \dftrans[x_1/x] \dfconcat \dftrans[x_2/x]$
\end{lemma}

\begin{proof}
We first prove the forward direction of the lemma above by structural induction over the syntax of $\dftrans$. Assume that $\dftrans[(x_1 \dfconcat x_2)/x] = \dftrans[x_1/x] \dfconcat \dftrans[x_2/x]$.

\bfpara{Base case: $\dftrans = x$.}
Trivially, applying rules \textsc{Var} and $\dfconcat$ gives $\dftrans[(x_1 \dfconcat x_2)/x] \horizdecompto \dftrans[x_1/x] \dfconcat \dftrans[x_2/x]$.

\bfpara{Inductive step: $\dftrans = \alpha(f, \dftrans')$.}
Given our induction hypothesis that $\dftrans'[(x_1 \dfconcat x_2)/x] \horizdecompto \dftrans'[x_1/x] \dfconcat \dftrans'[x_2/x]$, applying \textsc{Rel} yields
\[
\alpha(f, \dftrans'[(x_1 \dfconcat x_2)/x]) \horizdecompto \alpha(f, \dftrans'[x_1 / x]) \dfconcat \alpha(f, \dftrans'[x_2 / x]).
\]

Then, we prove the backward direction by proving a stronger variant: we argue that $\dftrans[(x_1 \dfconcat x_2)/x] = \dftrans[x_1/x] \dfconcat \dftrans[x_2/x]$ holds without assumption. We give a proof by structural induction over the syntax of $\dftrans$.

\bfpara{Base case: $\dftrans = x$.}
Then $\dftrans[(x_1 \dfconcat x_2)/x] = x_1 \dfconcat x_2 = \dftrans[x_1/x] \dfconcat \dftrans[x_2/x]$.

\bfpara{Inductive step: $\dftrans = \alpha(f, \dftrans')$.}
By IH and the definition of $\project$ and $\select$,
\begin{align*}
    \dftrans[(x_1 \dfconcat x_2)/x] = \alpha(f, \dftrans'[(x_1 \dfconcat x_2)/x]) = \alpha(f, \dftrans'[x_1/x]) \dfconcat \alpha(f, \dftrans'[x_2/x]) = \dftrans[x_1/x] \dfconcat \dftrans[x_2/x].
\end{align*}
\end{proof}
We now give the proof of \Cref{thm:soundness_completeness_hom_synthesis_rules}.

\bfpara{Forward direction.}
Assume $\dfagg = \abslambda{x}{\agg(\udaf, \init, \dftrans)} \horizdecompto h$ is a dataframe homomorphism producing a value of type $\type_r$ with merge operator $h$. By \Cref{thm:norm-complete}, $h$ is a normalizer, i.e., $\mathsf{Norm}(\type_r, \udaf, \init) = h$.
By \Cref{lem:soundness_and_completeness_horiz}, $\dftrans[(x_1\dfconcat x_2)/x] \horizdecompto \dftrans[x_1/x] \dfconcat \dftrans[x_2/x]$.
Finally, we apply the \textsc{Agg} and \textsc{Top} rule to obtain $\dfagg \horizdecompto h$.

\bfpara{Backward direction.}
Assume $\dfagg = \abslambda{x}{\agg(\udaf, \init, \dftrans)} \horizdecompto h$ for some function $h$. Let $x_1$ and $x_2$ be arbitrary dataframes. By \Cref{lem:soundness_and_completeness_horiz}, $\dftrans[(x_1\dfconcat x_2)/x] = \dftrans[x_1/x] \dfconcat \dftrans[x_2/x]$. Hence,
\begin{align*}
    \dfagg(x_1 \dfconcat x_2) &= \agg(\udaf, \init, \dftrans[(x_1\dfconcat x_2)/x]) \\
    &= \agg(\udaf, \init, \dftrans[x_1/x] \dfconcat \dftrans[x_2/x]).
\end{align*}
Also, note that $h$ is the merge operator of the program $\dfagg' = \abslambda{x}{\agg(\udaf, \init, x)}$ by \Cref{thm:norm-sound}, and we can derive:
\begin{align*}
    \dfagg(x_1 \dfconcat x_2) &= \agg(\udaf, \init, \dftrans[(x_1\dfconcat x_2)/x]) \\
    &= \agg(\udaf, \init, \dftrans[x_1/x] \dfconcat \dftrans[x_2/x]) \\
    &= h(\agg(\udaf, \init, \dftrans[x_1/x]), \agg(\udaf, \init, \dftrans[x_2/x])) \\
    &= h(\dfagg(x_1), \dfagg(x_2)).
\end{align*}
Thus, $\dfagg$ is a dataframe homomorphism with merge operator $h$.

\subsection{Proof of \Cref{thm:unrealizability}}

\bfpara{Proof of property (1).}
Assume there is a UDAF $\udaf: \tau_r \times \tau \to \tau_r$ and initializer $\init$ with a suitable normalizer $h$, i.e., satisfying $\forall s \in \type_r. \norm(s, \init) = s$, that violates property (1).
Let $(s, x)$ such that $f(\init, x) = \init$ and $f(s, x) \neq s$.
We have
\begin{align*}
    \udaf(s, x) = \udaf(\norm(s, \init), x) = \norm(s, \udaf(\init, x)) = \norm(s, \init) = s,
\end{align*}
which is a contradiction.

\bfpara{Proof of property (2).}
Assume there is a UDAF $\udaf: \tau_r \times \tau \to \tau_r$ and initializer $\init$ with a suitable normalizer $h$, i.e., satisfying $\forall s \in \type_r. \norm(s, \init) = s$, that violates property (2).
Let $(s, x)$ such that $f(\init, x) = f(\init, x')$ but $f(s, x) \neq f(s, x')$.
We have
\begin{align*}
    \udaf(s, x) &= \udaf(\norm(s, \init), x) = \norm(s, \udaf(\init, x)) \\
    &= \norm(s, \udaf(\init, x')) = \udaf(\norm(s, \init), x') \\
    &= \udaf(s, x'),
\end{align*}
which is a contradiction.

\subsection{Proof of \Cref{thm:norm-rules-soundness}}

\subsubsection{Forward direction: Completeness}

Let $\udaf : \type_r \rightarrow \type \rightarrow \type_r$ be an arbitrary UDAF. We show that if $\norm$ is a normalizer of $\udaf$, then we have $\normalizes{\udaf}{\init}{\norm}$. Since $\norm$ is a normalizer, it satisfies $\Phi_1$ and $\Phi_2$ used in the \textsc{Norm-Synth} rule. Thus, we resort to the completeness of SyGuS solver used in \textsf{Solve}.

\subsubsection{Backward direction: Soundness}

We give a proof by structural induction on the rules in \Cref{fig:witness_construction}.
Let $\udaf : \type_r \rightarrow \type \rightarrow \type_r$ be an arbitrary UDAF. We show that if $\normalizes{\udaf}{\init}{\norm}$, then $\norm$ is a normalizer of $f$ satisfying  $\forall s. ~h(s, \init) = s$.

\bfpara{Base case: \textsc{Norm-Synth}.}
Since $\Phi_1 \land \Phi_2$ precisely encodes the commutativity constraint of normalizers and $\forall s. ~h(s, \init) = s$, the 
condition mentioned in the theorem, we resort to the soundness of SyGuS solver used in \textsf{Solve}.

\bfpara{Inductive step: \textsc{Norm-Tuple}.}
Assume that for every $i$, we have $\normalizes{\udaf_i}{\typeselector_i(\init)}{h_i}$ and $h_i$ is a normalizer of $f_i$ satisfying $\forall s. ~h_i(s, \typeselector_i(\init)) = s$. Let $s, a, b \in \type_r$ and $x \in \type$.
Then, we have
\begin{align*}
    h(s, \init) &= (h_1(\typeselector_1(s), \typeselector_1(\init)), \ldots, h_n(\typeselector_n(s), \typeselector_n(\init))) \\
    &= (\typeselector_1(s), \ldots, \typeselector_n(s)) = s,
\end{align*}
and
\begin{align*}
    h(a, f(b, x)) &= h(a, (f_1(\typeselector_1(b), x), \ldots, f_n(\typeselector_n(b), x))) \\
    &= \left( h_1(\typeselector_1(a), f_1(\typeselector_1(b), x)), \ldots, h_n(\typeselector_n(a), f_n(\typeselector_n(b), x)) \right) \\
    &= \left( f_1(h_1(\typeselector_1(a), \typeselector_1(b)), x), \ldots, f_n(h_n(\typeselector_n(a), \typeselector_n(b)), x) \right) \\
    &= f((h_1(\typeselector_1(a), \typeselector_1(b)), \ldots, h_n(\typeselector_n(a), \typeselector_n(b))), x) \\
    &= f(h(a, b), x).
\end{align*}

\bfpara{Inductive step: \textsc{Norm-Coll}.}
Assume that $\normalizes{\udaf'}{\default{\type}}{h}$ and $h$ is a normalizer of $\udaf'$ satisfying $\forall s. ~h_i(s, \default{\type}) = s$. Let $s, a, b \in \typecollection{\type}$ and $x \in \type$. Let $\init = \default{\typecollection{\type}}$.
Then, we have
\begin{align*}
    h(s, \init) &= \convert_{\typecollection{\type}}(\set{(k, h(v, v_i)) ~|~ k, v, v_i \in \convert_\textsf{Map}(s) \collectionjoin \convert_\textsf{Map}(\init)}) \\
    &= \convert_{\typecollection{\type}}(\set{(k, h(v, \default{\type})) ~|~ k, v, v_i \in \convert_\textsf{Map}(s) \collectionjoin \convert_\textsf{Map}(\init)}) \\
    &= \convert_{\typecollection{\type}}(\set{(k, v) ~|~ \convert_\textsf{Map}(s) \collectionjoin \convert_\textsf{Map}(\init)}) = s,
\end{align*}
and
\begin{align*}
h(a, f(b, x)) &= \convert_{\typecollection{\type}}(\set{(k, h(v_1, v_2)) ~|~ k, v_1, v_2 \in \convert_\textsf{Map}(a) \collectionjoin \convert_\textsf{Map}(f(b, x))}) \\
&= \convert_{\typecollection{\type}}(\set{(k, h(v_1, f'(v_b, x))) ~|~ p(v_b), ~k, v_1, v_b \in \convert_\textsf{Map}(a) \collectionjoin \convert_\textsf{Map}(b)}) \\
&= \convert_{\typecollection{\type}}(\set{(k, f'(h(v_1, v_b), x)) ~|~ p(v_b), ~k, v_1, v_b \in \convert_\textsf{Map}(a) \collectionjoin \convert_\textsf{Map}(b)}) \\
&= f(h(a, b), x).
\end{align*}

\subsection{Proof of \Cref{thm:sound-simplification}}

\begin{lemma}\label{lemma:pp-soundness}
    Suppose $(x: \type, \typedestructor, \decomp) \propagateparamto \decomp'$. If $\decomp \decomptoexpr E$ and $\decomp' \decomptoexpr E'$, then $E'$ and $\abslambda{x}{E[\typedestructor(x)/x]}$ is semantically equivalent. 
\end{lemma}

\begin{proof}
We proceed the proof by structural induction on the rules in \Cref{fig:pp}.

\bfpara{Base case: \textsc{Expr}.}
Suppose $(x: \type, \typedestructor, \decomp) \propagateparamto \vdcomp{\abslambda{(x: \type)}{E}}{\typedestructor}$.
We apply the \textsc{Abs-Base} rule to get
$\vdcomp{\abslambda{(x: \type)}{E}}{\typedestructor} \decomptoexpr \abslambda{x}{E[\typedestructor(x)/x]}$.

\bfpara{Inductive step: \textsc{Function}.}
Suppose $(x: \type, \typedestructor, \isf) \propagateparamto \vdcomp{\abslambda{(x: \type)}{\isf}}{\typedestructor}$.
We apply the \textsc{Abs-Ind} rule, and by IH, we have
\(
\vdcomp{\abslambda{(x: \type)}{\isf}}{\typedestructor} \decomptoexpr \abslambda{x}{E[\typedestructor(x)/x]},
\)
where $\isf \decomptoexpr E$.

\bfpara{Inductive step: \textsc{Tuple-Base}.}
Suppose
\(
(x: \typebase, \typedestructor, \bounded{\vdfunc}{\decomp_1, \ldots, \decomp_n})
\propagateparamto
\bounded{\vdfunc}{\decomp_1', \ldots, \decomp_n'}.
\)
Assume that for every $i$, $\decomp_i \decomptoexpr E_i$ and $\decomp_i' \decomptoexpr E_i' = \abslambda{x}{E_i[d(x)/x]}$.
We apply the \textsc{Tuple} rule, and by IH, we have
\begin{align*}
\bounded{\vdfunc}{\decomp_1', \ldots, \decomp_n'} &\decomptoexpr \abslambda{x}{(E_1[d(x)/x], \ldots, E_n[d(x)/x])} \\
&= \abslambda{x}{(E_1, \ldots, E_n)[d(x)/x]} \\
&= \abslambda{x}{E[d(x)/x]}.
\end{align*}

\bfpara{Inductive step: \textsc{Tuple-Inductive}.}
Suppose
\(
(x: (\tau_1, \ldots, \tau_n)), \typedestructor, \bounded{\vdfunc}{\decomp_1, \ldots, \decomp_n})
\propagateparamto
\bounded{\vdfunc}{\decomp_1', \ldots, \decomp_n'}.
\)
Assume that for every $i$, $\decomp_i \decomptoexpr E_i$, $\bar{\decomp}_i \decomptoexpr \bar{E}_i$, and $\decomp_i' \decomptoexpr E_i' = \abslambda{x}{\bar{E}_i[\typeselector_i(d(x))/x]}$.
We apply the \textsc{Tuple} rule, and by IH, we have
\begin{align*}
\bounded{\vdfunc}{\decomp_1', \ldots, \decomp_n'}
&\decomptoexpr \abslambda{x}{(\bar{E}_1[\typeselector_1(d(x))/x], \ldots, \bar{E}_n[\typeselector_n(d(x))/x])} \\
&= \abslambda{x}{(E_1[d(x)/x], \ldots, E_n[d(x)/x])} \\
&= \abslambda{x}{(E_1, \ldots, E_n)[d(x)/x]} \\
&= \abslambda{x}{E[d(x)/x]}.
\end{align*}

\bfpara{Inductive step: \textsc{C-Base}.}
Suppose
\(
(x: \typebase, \typedestructor, \unbounded{\typeiter}{\decomp}{\vdpred}{\typedestructor_0})
\propagateparamto
\unbounded{\typeiter}{\decomp'}{\vdpred}{\typedestructor_0}.
\)
Assume that $\decomp \decomptoexpr \hat{E}$ and $\decomp' \decomptoexpr \hat{E}' = \abslambda{x}{\hat{E}[d(x)/x]}$.
Let $E$ be such that $\unbounded{\typeiter}{\decomp}{\vdpred}{\typedestructor_0} \decomptoexpr E$.
We apply the \textsc{C2} rule, and by IH, we have
\begin{align*}
\unbounded{\typeiter}{\decomp'}{\vdpred}{\typedestructor_0}
&\decomptoexpr \abslambda{x}{\abslambda{Y}{\abslambda{\bar{x}}{ \map(\hat{E}'(y, \bar{x}), \filter(\vdpred, d(Y))) }}} \\
&= \abslambda{x}{\abslambda{Y}{\abslambda{\bar{x}}{ \map(\hat{E}[d(x)/x](y, \bar{x}), \filter(\vdpred, d(Y))) }}} \\
&= \abslambda{x}{\abslambda{Y}{\abslambda{\bar{x}}{ \map(\hat{E}(y, \bar{x}), \filter(\vdpred, d(Y)))[d(x)/x] }}} \\
&= \abslambda{x}{E[d(x)/x]}.
\end{align*}

\bfpara{Inductive step: \textsc{C-Ind}.}
Suppose
\(
(X: \typecollection{\type}, \typedestructor, \unboundedapp{\typeiter}{\decomp}{\vdpred}{x \in X})
\propagateparamto
\unbounded{\typeiter}{\decomp'}{\abslambda{x}{\vdpred}}{\typedestructor}.
\)
Assume that $\decomp \decomptoexpr \hat{E}$ and $\decomp' \decomptoexpr \hat{E}' = \abslambda{x}{\hat{E}[d(x)/x]}$.

Let $E$ be such that $\unboundedapp{\typeiter}{\decomp}{\vdpred}{x \in X} \decomptoexpr E$.
Note that
\begin{align*}
E &= \abslambda{\bar{x}}{\map(\hat{E}(\bar{x}), \filter(\abslambda{x}{\vdpred}, X))}.
\end{align*}

We apply the \textsc{C2} rule, and by IH, we have
\begin{align*}
\unbounded{\typeiter}{\decomp'}{\abslambda{x}{\vdpred}}{\typedestructor}
&\decomptoexpr \abslambda{Y}{\abslambda{\bar{x}}{ \map(\hat{E}'(y, \bar{x}), \filter(\abslambda{x}{\vdpred}, d(Y))) }} \\
&= \abslambda{Y}{\abslambda{\bar{x}}{ \map(\hat{E}[d(x)/x](\bar{x}), \filter(\abslambda{x}{\vdpred}, d(Y))) }} \\
&= \abslambda{Y}{\abslambda{\bar{x}}{ \map(\hat{E}(\bar{x}), \filter(\abslambda{x}{\vdpred}, X))[d(Y)/X] }} \\
&= \abslambda{Y}{E[d(Y)/X]}.
\end{align*}

\end{proof}

We now prove \Cref{thm:sound-simplification} using structural induction on the rules in \Cref{fig:vertical_decomp}.
Let $E$ be an arbitrary expression.

\bfpara{Base case: \textsc{BaseType}.}
Assume $E : \typebase$. We have $E \vertdecompto E$, and applying \textsc{Expr} gives $E \decomptoexpr E$. It is straightforward that $E = E$.

\bfpara{Base case: \textsc{Lam-Base}.}
Assume $E$ is a built-in function $f$. Similarly, we have $f \vertdecompto f$, and we apply \textsc{Expr} to get $f \decomptoexpr f$.

\bfpara{Inductive step: \textsc{Tuple}.}
Suppose $E : \typeprod = (E_1, \ldots, E_n) \vertdecompto \bounded{\vdfunc}{\decomp_1, \ldots, \decomp_n}$.
Assume that for every $i$, we have $E_i \vertdecompto \decomp_i$ and $\decomp_i \decomptoexpr E_i'$ where $E_i$ and $E_i'$ are semantically equivalent.
We apply the \textsc{Tuple} rule and get $\bounded{\vdfunc}{\decomp_1, \ldots, \decomp_n} \decomptoexpr (E_1, \ldots, E_n) = E$.

\bfpara{Inductive step: \textsc{Collection}.}
Assume $E = X$, a collection-typed identifier. Then $E \vertdecompto \unboundedapp{\typeiter}{x}{\predtrue}{x \in X}$.
Applying \textsc{C1} yields
\(
\unboundedapp{\typeiter}{x}{\predtrue}{x \in X} \decomptoexpr
\map(\funcid, \filter(\top, X)) = X.
\)

\bfpara{Inductive step: \textsc{Lam-Ind}.}
Suppose $E = \abslambda{x:\type}{e} \vertdecompto \decomp'$ and the premise in \textsc{Lam-Ind} holds.
Assume that IH holds, namely $\decomp \decomptoexpr e'$ where $e$ and $e'$ are semantically equivalent.
By \Cref{lemma:pp-soundness}, $\decomp' \decomptoexpr E'$ is semantically equivalent to $\abslambda{x}{e'[\funcid(x)/x]} = \abslambda{x}{e'} = \abslambda{x}{e}$.

\bfpara{Inductive step: \textsc{Map}.}
Suppose $E = \map(f, e) \vertdecompto \unboundedapp{\typeiter}{\decomp'}{\vdpred}{x \in X}$ and the premise in \textsc{Map} holds. By IH, we apply \textsc{C1} and get
\begin{align*}
    \unboundedapp{\typeiter}{\decomp'}{\vdpred}{x \in X} &\decomptoexpr \map(\abslambda{x}{f(e_v)}, \filter(\abslambda{x}{\vdpred}, X)) \\
    &= \map(f, \map(\abslambda{x}{e_v}, \filter(\abslambda{x}{\vdpred}, X))) \\
    &= \map(f, E).
\end{align*}

\bfpara{Inductive step: \textsc{Filter}.}
Suppose $E = \map(f, e) \vertdecompto \unboundedapp{\footnotesize\typeiter}{\decomp}{\vdpred \land p(e_v)}{x \in X}$ and the premise in \textsc{Filter} holds. By IH, we apply \textsc{C1} and get
\begin{align*}
    \unboundedapp{\footnotesize\typeiter}{\decomp}{\vdpred \land p(e_v)}{x \in X}
    &\decomptoexpr \map(\abslambda{x}{e_v}, \filter(\abslambda{x}{\vdpred \land p(e_v)}, X)) \\
    &= \filter(p, \map(\abslambda{x}{e_v}, \filter(\abslambda{x}{\vdpred}, X))) \\
    &= \filter(p, e).
\end{align*}

\subsection{Proof of \Cref{thm:completeness_of_decomposition}}

Let $h$ be the return value of \textsf{IsHomomorphism}($\dfagg$) where $\dfagg$ is a surjective function from $\mathsf{DF}\langle \tau \rangle$ to set $ \tau_r$.
First, we show that if $h = \bot$, then $\dfagg$ is not a dataframe homomorphism.
\Cref{lst:top-level:refute1,lst:top-level:refute2} are the only places where \textsf{IsHomomorphism} returns $\bot$. We discuss both cases below.
\begin{enumerate}
    \item \Cref{lst:top-level:refute1} returns $\bot$: by \Cref{thm:soundness_completeness_hom_synthesis_rules}, $\dfagg$ is not a homomorphism.
    \item \Cref{lst:top-level:refute2} returns $\bot$: by \Cref{thm:unrealizability}, $\dfagg$ is not a homomorphism.
\end{enumerate}

Next, we show that if $h \neq \bot$, then $\dfagg$ is a dataframe homomorphism with merge operator $h$.
Similarly, \Cref{lst:top-level:step2,lst:top-level:step4,lst:top-level:step5} are the places where \textsf{IsHomomorphism} returns $h \neq \bot$. We discuss these cases below.

\begin{enumerate}
    \item \Cref{lst:top-level:step2,lst:top-level:step4} return $h$: by \Cref{thm:norm-rules-soundness}, $h$ is the merge operator of $\dfagg$ and $\dfagg$ is a dataframe homomorphism.
    \item \Cref{lst:top-level:step5} returns $h$: by \Cref{thm:sound-simplification}, the decomposed expression $\decomp$ on \cref{lst:top-level:decompose} is semantically equivalent to $\udaf$. Then, by \Cref{thm:norm-rules-soundness}, $h$ is the merge operator of $\dfagg$ and $\dfagg$ is a dataframe homomorphism. 
\end{enumerate}

%% file: appendix.tex
\section{SyGuS Encoding for Leaf-Level Synthesis}\label{appendix:sygus}

This section provides a detailed description of our approach to encoding the leaf-level synthesis problems for the SyGuS solver, CVC5. These problems are generated by the{\sc Norm-Synth} rule in \Cref{fig:witness_construction}.

\subsection{The Synthesis Task}

For each leaf-level synthesis problem, the goal is to synthesize a normalizer function, $h$, that takes two accumulator states as input and returns a new, merged state. The specification for $h$ is derived directly from the theoretical requirements of a normalizer, encoded as two primary constraints for the solver: $\Phi_1$, the initializer side condition, and $\Phi_2$,  the commutativity condition specified in \Cref{def:commute}.

\subsection{Grammar for Primitive Types}

The foundation of the SyGuS grammar consists of productions for primitive types like integers, booleans, and strings.

\bfpara{Terminals.} The grammar's terminals include the input variables for the normalizer function (representing the two accumulator states), common constants such as \texttt{0}, \texttt{1}, \texttt{true}, \texttt{false}, and any other constants extracted from the UDAF's body. To aid in reasoning about aggregations involving \texttt{min} or \texttt{max}, we also include symbolic constants that represent boundary values, such as the minimum and maximum possible values for a given numeric type.

\bfpara{Operations.} The grammar includes standard unary and binary operations for primitive types.

\begin{enumerate}
    \item For integers, this includes arithmetic operators (\texttt{+}, \texttt{-}, \texttt{*}, \texttt{/}) and comparisons (\texttt{<}, \texttt{>}, \texttt{==}).
    \item For booleans, it includes logical operators (\texttt{and}, \texttt{or}, \texttt{not}).
    \item For strings, the grammar incorporates primitive operators from CVC5's theory of strings, such as \texttt{str.++} (concatenation) and \texttt{str.len}.
\end{enumerate}

The grammar is also supplied with the \texttt{ITE} operator to allow for the synthesis of conditional expressions.

\subsection{Grammar for Collection Types}

A key aspect of our encoding is the handling of complex data structures. Our strategy was determined empirically: we found that mapping our DSL's collections to the native types in CVC5's extended theories consistently outperformed alternative encodings, such as using a general theory of Abstract Data Types.

\bfpara{Lists, Sets, and Tuples.}
Following our empirical findings, we encode lists as CVC5 sequences, and sets and tuples using their corresponding native solver types. This approach allows us to leverage CVC5’s powerful, built-in theories for these structures. The grammar is populated with the solver's native operations, such as \texttt{seq.concat} and \texttt{seq.len} for lists, \texttt{set.union} and \texttt{set.insert} for sets, and \texttt{tuple.select} for tuples.

\bfpara{Maps.}
In contrast, map primitives are not natively supported by the solver. To handle them, we model maps as a set of key-value pairs (e.g., \texttt{(Set (Tuple Key Value))}). We then provide a custom library of primitive map operations, implemented as recursive functions in the SyGuS format. This library includes essential functions like \texttt{map.access}, \texttt{map.update}, \texttt{map.contains\_key}, and \texttt{map.map\_values}. This approach allows us to reason about maps while staying within the well-supported theory of sets.

\bfpara{Combined Operations.}
The grammar also supports more complex, compositional patterns. For example, it provides the ability to \textit{zip} two lists into a list of pairs, which can then be transformed using a \texttt{map} operation, a common pattern in UDAF merge logic.

\subsection{Grammar Ordering and Solver Heuristics}

The sensitivity to the order of non-terminals arises because of incomplete heuristics used by the underlying SyGuS solver (CVC5). Different orderings of grammar non-terminals can affect the solver's internal search paths, even though they don't alter the problem's semantics or solution space.

In our implementation, we chose a simple alphabetical ordering for grammar non-terminals rather than tuning it for specific benchmarks. Our primary goal was to emphasize a principled decomposition strategy, making leaf-level synthesis tasks simpler and more robust, rather than relying on solver-specific heuristics. Optimizing grammar ordering on a per-benchmark basis would have introduced undesirable complexity and potentially reduced the generality of our solution.